\documentclass[a4paper,
               DIV=10, 
               titlepage=off]{scrartcl}

\usepackage[utf8]{inputenc}
\usepackage[english]{babel}
\usepackage{csquotes}

\usepackage{amsthm,amsmath,amsfonts,amssymb}
\usepackage{bbm,nicefrac}
\usepackage{setspace,adjustbox,array}
\newcolumntype{L}{>{\centering\arraybackslash}m{3cm}}

\PassOptionsToPackage{breaklinks}{hyperref}
\usepackage[colorlinks=true,linkcolor=blue,urlcolor=blue,citecolor=blue,anchorcolor=green,pdfusetitle]{hyperref}
\usepackage{dsfont,xspace,xparse}
\usepackage{lmodern,newtxtt,microtype}
\usepackage{tabularx,xifthen,afterpage,booktabs}
\usepackage{qcircuit,algpseudocode,algorithm}

\numberwithin{equation}{section}

\usepackage{microtype,parskip} % better hyphenation, neater typography
\setparsizes{1.5em}{0.35\baselineskip plus 0.5\baselineskip}{1em plus 1fil} % para indents

\usepackage[dvipsnames]{xcolor}

\usepackage{upgreek}

\usepackage[capitalize,nameinlink]{cleveref}
\crefname{lemma}{Lemma}{Lemmas}
\crefname{proposition}{Proposition}{Propositions}
\crefname{definition}{Definition}{Definitions}
\crefname{theorem}{Theorem}{Theorems}
\crefname{conjecture}{Conjecture}{Conjectures}
\crefname{corollary}{Corollary}{Corollaries}
\crefname{example}{Example}{Examples}
\crefname{section}{Section}{Sections}
\crefname{appendix}{Appendix}{Appendices}
\crefname{figure}{Fig.}{Figs.}
\crefname{equation}{Eq.}{Eqs.}
\crefname{table}{Table}{Tables}
\crefname{item}{Property}{Properties}
\crefname{remark}{Remark}{Remarks}
\crefname{problem}{}{}

\newtheorem{theorem}{Theorem}
\newtheorem{definition}{Definition}

\newtheorem{lemma}{Lemma}

\newtheorem*{remark*}{Remark}

\renewenvironment{abstract}
 { \small
  \list{}{%\textbf{Abstract:}
    \setlength{\leftmargin}{6mm}
    \setlength{\rightmargin}{6mm}%
  }%
  \item\relax}
 {\endlist}

\newcommand\prob\textsc
\makeatletter
\newcommand{\problemtitle}[1]{\gdef\@problemtitle{#1}}
\newcommand{\probleminput}[1]{\gdef\@probleminput{#1}}
\newcommand{\problemquestion}[1]{\gdef\@problemquestion{#1}}
\newcommand{\problempromise}[1]{\gdef\@problempromise{#1}}
\DeclareDocumentEnvironment{problem}{o o}{
\problemtitle{\IfValueT{#1}{#1}}\probleminput{}\problempromise{}\problemquestion{}
}{
  \par\addvspace{.5\baselineskip}
  \noindent
  \begin{tabularx}{\textwidth}{@{\hspace{\parindent}} l X c}
    \multicolumn{2}{@{\hspace{\parindent}}l}{\prob{\@problemtitle}} \\% Title
    \textbf{Input:} & \@probleminput \\% Input
\ifx\@problempromise\empty%
\else%
    \textbf{Promise:} & \@problempromise \\% Input
\fi%
    \textbf{Question:} & \@problemquestion% Question
  \end{tabularx}
}
\makeatother

\usepackage{listings}
\usepackage[backend=biber,
            style=trad-alpha,
            sorting=nty, %nty
            doi=false,
            isbn=false,
            url=false,
            maxbibnames=20,
            maxcitenames=2, 
            backref,
            backrefstyle=three]{biblatex}  

\DefineBibliographyStrings{english}{%
  backrefpage = {page},% originally "cited on page"
  backrefpages = {pages},% originally "cited on pages"
}
\newbibmacro{string+doi}[1]{\iffieldundef{doi}{#1}{\href{https://dx.doi.org/\thefield{doi}}{#1}}}
\DeclareFieldFormat{title}{\usebibmacro{string+doi}{\mkbibemph{#1}}}
\DeclareFieldFormat[article]{title}{\usebibmacro{string+doi}{\mkbibquote{#1}}}
\DeclareFieldFormat[incollection]{title}{\usebibmacro{string+doi}{\mkbibquote{#1}}}
\DeclareFieldFormat[inproceedings]{title}{\usebibmacro{string+doi}{\mkbibquote{#1}}}
\renewbibmacro{in:}{} % environments, code listings, biblatex, cleveref etc. 

\newcommand{\C}{{\mathbb{C}}} 
\newcommand{\R}{{\mathbb{R}}}
 
\newcommand{\N}{{\mathbb{N}}}
\newcommand{\id}{{\mathbbm{1}}} % identity matrix 
\newcommand{\ee}{\mathrm{e}}
\newcommand{\ii}{\mathrm{i}}
\newcommand{\ww}{\mathrm{w}} % probability weight

\newcommand{\CNOT}{\mathop\text{CNOT}}     % CNOT gate

\newcommand{\pbold}{\mathbf{p}}

\renewcommand{\O}{\mathcal{O}}
\newcommand{\norm}[1]{\|{#1}\|}

\newcommand{\Abs}[1]{\left|{#1}\right|}
\newcommand{\tr}{\mathop{}\text{Tr}}
\newcommand{\Tr}[2][]{\tr{%
	\ifthenelse{\isempty{#1}}{}{_{#1}}}%
	\left({#2}\right)}
	
\newcommand{\ket}[1]{\left | #1 \right \rangle}
\newcommand{\bra}[1]{\left \langle #1 \right |}
\newcommand{\braket}[1]{\left\langle #1 \right\rangle}
\newcommand{\ketbra}[2]{\ket{#1}\!\!\bra{#2}}
\newcommand{\proj}[1]{\ketbra{{#1}}{{#1}}}

\DeclareMathOperator{\poly}{poly}
\newcommand*{\defeq}{:=}

\newcommand{\parahead}[1]{\noindent\textbf{{#1}}}

\title{Sublinear quantum algorithms for estimating von Neumann entropy}

\author{
\large{Tom Gur \thanks{University of Warwick Email: \texttt{tom.gur@warwick.ac.uk}.}}
\and 
\large{Min-Hsiu Hsieh} \thanks{Hon Hai (Foxconn) Quantum Computing Research Center. Email: \texttt{min-hsiu.hsieh@foxconn.com}.}
\and 
\large{Sathyawageeswar Subramanian} \thanks{University of Warwick. Email: \texttt{sathya.subramanian@warwick.ac.uk}.\newline\noindent T.\ Gur and S.\ Subramanian are supported by the UKRI Future Leaders Fellowship MR/S031545/1.}
}
\date{}

\begin{document}

\maketitle
%%%%%%%%%%%%%%%%%%%%%%%%%%%%%%%%%%%%%%%%%%%%%%%%%%%%%%%%%%%%%%%%%%%%%%%%%%%%%%%%%%%%%%%%%%

\begin{abstract}
\begin{center}
    \textbf{Abstract}
\end{center}
    Entropy is a fundamental property of both classical and quantum systems, spanning myriad theoretical and practical applications in physics and computer science.
    We study the problem of obtaining estimates to within a multiplicative factor $\gamma>1$ of the Shannon entropy of probability distributions and the von Neumann entropy of mixed quantum states. Our main results are:
    
    \begin{itemize}
        \item  an $\widetilde{\O}\left( n^{\frac{1+\eta}{2\gamma^2}}\right)$-query quantum algorithm that outputs a $\gamma$-multiplicative approximation of the Shannon entropy $H(\pbold)$ of a classical probability distribution $\pbold = (p_1,\ldots,p_n)$; 
        
        \item  an $\widetilde{\O}\left( n^{\frac12+\frac{1+\eta}{2\gamma^2}}\right)$-query quantum algorithm that outputs a $\gamma$-multiplicative approximation of the von Neumann entropy $S(\rho)$ of a density matrix $\rho\in\C^{n\times n}$.
    \end{itemize}
    In both cases, the input is assumed to have entropy bounded away from zero by a quantity determined by the parameter $\eta>0$, since, as we prove, no polynomial query algorithm can multiplicatively approximate the entropy of distributions with arbitrarily low entropy. In addition, we provide $\Omega\left(n^{\nicefrac{1}{3\gamma^2}}\right)$ lower bounds on the query complexity of $\gamma$-multiplicative estimation of Shannon and von Neumann entropies.
    We work with the quantum purified query access model, which can handle both classical probability distributions and mixed quantum states, and is the most general input model considered in the literature.
\end{abstract}

\clearpage
\tableofcontents
\clearpage

%%%%%%%%%%%%%%%%%%%%%%%%%%%%%%%%%%%%%%
\section{Introduction}
\label{sec:intro}
%%%%%%%%%%%%%%
Entropy as a scientific concept is central in the study of a vast variety of subjects, ranging from thermodynamics and information theory to networks and quantum entanglement. The notion of entropy mathematically measures the amount of disorder and uncertainty in a system composed of many parts. Indeed, the second law of thermodynamics can famously be encapsulated in the simple statement that the entropy of a closed system can never decrease.

In this work we focus on the most fundamental entropic functionals for both classical and quantum objects. For classical systems, the Shannon entropy $H(\pbold) :=\linebreak[4] -\sum_{i\in[n]} p_i \log p_i$ of a probability distribution $\pbold = (p_1,\ldots,p_n)$ is a cornerstone of information theory. Notably, the Shannon entropy is proportional to the rates at which input data can be transmitted over communication channels \cite{shannon}. % pesky lineabreak[] for the shannon entropy eq to stop arxiv from making the eq go out of the margin

For quantum systems, the von Neumann entropy $S(\rho):=-\Tr{\rho\log\rho}$ \cite{Neumann1927,PetzOnNeumann} of a mixed state specified by its density matrix $\rho\in\C^{n\times n}$ is pivotal to our understanding of key properties in quantum mechanics, such as the amount of entanglement contained in bipartite quantum systems. In terms of information theory, von Neumann entropy gives an asymptotic lower bound for the rate at which quantum data can be compressed in a noiseless fashion \cite{PhysRevA.51.2738}. 

The von Neumann entropy is a strict generalisation of the Shannon entropy and reduces to the latter when viewed in appropriately restricted settings. Other entropic functionals built on the von Neumann entropy are also widely used in characterising quantum systems. Moreover, they arise extensively in condensed matter and high energy physics \cite{Laflorencie2016}, and are often used as operational measures in quantum information-processing tasks \cite{Konig2009,WinterCoherent}. They have also had immense theoretical implications in the theory of gravity and black holes \cite{Bekenstein1973,Dong2016}, and their study from a quantum information-theoretic viewpoint continues to be a rich source of new physical insights \cite{Azuma18,Azuma20}.

Since obtaining a perfect description of a system is typically impossible, estimating the entropy of an unknown probability distribution or quantum state using a bounded number of samples, queries, or measurements is a vital algorithmic task. This question received much attention in classical information theory, and in a series of works \cite{Batu2002,WuYang15additive,Jiao15additive,ValiantValiant2011} spanning the last two decades, nearly tight bounds were shown on the sample complexity of classical algorithms for estimating Shannon entropy.

In this paper, we study quantum algorithms for the problem of obtaining multiplicative estimates of the Shannon entropy of classical probability distributions and the von Neumann entropy of mixed quantum states. For generality, we focus on algorithms in the quantum purified query access model, which can handle both classical probability distributions and mixed quantum states, and is the most general input model considered in the literature. Thus, our algorithms can be implemented in all the four major input models for quantum algorithms accessing probability distributions: quantum samples, quantum queries to frequency vectors and classically drawn samples, as well as purifications.

Approximation algorithms that estimate a quantity to within a multiplicative factor, i.e., estimating $x>0$ by outputting $\tilde{x} \in [\nicefrac{x}{\gamma}, \gamma x]$ for some $\gamma>1$, allow for much flexibility. On the one hand, via a correct choice of parameters they allow us to recover additive approximations, i.e., estimating $x>0$ by outputting $\tilde{x} \in [x-\epsilon,x+\epsilon]$ for a small precision parameter $\epsilon\in(0,1)$ (see \cref{sec:applications}). On the other hand, in natural settings of parameters, they also allow for a greater slack than additive approximation (for instance, in many applications we may only need to know the unknown quantity to within a factor of two, i.e., $\gamma=2$). This slack allows us to obtain far more efficient algorithms than is possible with additive approximation; indeed, we attain sublinear complexity for entropy estimation as opposed to polynomial complexity (which is the best that can be achieved for additive estimates). In turn, these properties of multiplicative approximation algorithms make them often more involved and harder to construct.
%Approximation algorithms that estimate a quantity to within a multiplicative factor, i.e., estimating $x>0$ by outputting $\tilde{x} \in [\nicefrac{x}{\gamma}, \gamma x]$ for some $\gamma>1$, are particularly valuable when the target quantity might be small. In contrast to the typical situation of additive approximation with a small precision $\epsilon\in(0,1)$, multiplicative approximation allows for greater slack than additive approximation --- for instance, in many applications we may only need to know the unknown quantity to within a factor of two, i.e., $\gamma=2$. This slack allows us to obtain far more efficient algorithms than is possible with additive approximation; indeed, we attain sublinear complexity for entropy estimation as opposed to polynomial complexity. Furthermore, by a correct setting of parameters we can also recover the stricter additive approximations, with a complexity matching the known bounds. However, although multiplicative approximation algorithms may have lower complexity than additive ones, they are usually more involved and hard to construct.

In the classical setting, \citeauthor{Batu2002} \cite{Batu2002} considered the problem of estimating Shannon entropy to multiplicative precision, showing that $\widetilde{\O}(n^{\nicefrac{(1+\eta)}{\gamma^2}})$ samples suffice to obtain an estimate within a factor $\gamma$, for classical distributions $p$ with $H(p)>\nicefrac{\gamma}{\eta}$. This is almost matched by a lower bound of $\Omega(n^{\nicefrac{(1-\eta)}{\gamma^2}})$ later proven in \cite{Valiant2011}. 

We build on the aforementioned line of work by extending it to the setting of quantum query algorithms for both classical distributions and mixed quantum states. In particular, we are motivated by the following question:

\begin{center}
    \emph{Is it possible to construct sublinear quantum algorithms\\ for estimating von Neumann entropy?}
\end{center}
%%%%%%%%%%%%%%

%%%%%%%%%%%%%%%%%%%%%%%%%%%%%%%%%%%%%%%%%%%%%%%%%%%%%%%%%%
\subsection{Main results}
This paper answers the foregoing question in the affirmative. We begin by presenting our construction of quantum query algorithms for estimating the Shannon entropy of a probability distribution to within a multiplicative factor $\gamma>1$. Then, we proceed to estimating the von Neumann entropy of mixed quantum states. Finally, we show lower bounds on the quantum query complexity of the foregoing problems. The query complexity in our setting is the standard analogue to the classical sample complexity, and our results are in this sense information theoretic in nature. 
%%%%%%%%%%%%%%%%%%%%%%%%%%%%%%%%%%%%%%%
\subsubsection{Estimating Shannon entropy}
We first consider quantum algorithms for classical distributions, accessed via quantum query oracles. We obtain a quadratic improvement in the information theoretic complexity over the best possible classical algorithm.
\begin{theorem}
\label{thm:class}
    There is a quantum algorithm that outputs with high probability a $\gamma$-multiplicative approximation of the Shannon entropy $H(\pbold)$ of a classical probability distribution $\pbold$ on $[n]$ accessed via a purified quantum query oracle $U_{\pbold}$ as in \cref{eq:purified-access-class} using $U_{\pbold}$ and its inverse $\widetilde{\O}\left( n^{\nicefrac{(1+\eta)}{2\gamma^2}}\right)$ times, provided $H(\pbold)>3\gamma+\nicefrac{4}{\eta}$. 
\end{theorem}

We remark that there are four popular quantum input models that have been studied for quantum algorithms accessing classical input probability distributions \cite{Bravyi2011,chakraborty18Testing,Montanaro2013,Li2019RenyiQuery,Belovs2019,Gilyen2019Distributional}, namely:
(i) quantum query oracle to a frequency vector,
(ii) quantum query oracle to list of classically generated samples, 
(iii) quantum samples with preparation oracle, and
(iv) purified quantum query access.
(See formal definitions in \cref{sec:prelim-ests}.)
We stress that \cref{thm:class} holds for all of the models above.

In more detail, Belovs \cite{Belovs2019} initiated a comparative study of the aforementioned four models and showed that the purified access model (i.e., Model (iv)) is the most general in the sense that it can capture both classical distributions and density matrices, as well as be emulated by all the other models with a constant overhead. Furthermore, he proved that the quantum samples model  (i.e., Model (iii)) is strictly stronger than the rest of the models, and conjectured that Models (i), (ii), and (iv) are equivalent for classical probability distributions. Our algorithms are constructed using purified query oracles, and hence by the above, the upper bounds we prove automatically apply to the rest of the models. 

%and proved some preliminary results indicating that models (1), (2) and (4) are equivalent for classical probability distributions, while (3) is strictly stronger.
 %Furthermore, all the algorithms known to us in model (1), which happens to be the most popular and well-studied model for distributional property testing using quantum algorithms, proceed by first generating an oracle in model (4) using the model (1) oracle. 

%%%%%%%%%%%%%%%%%%%%%%%%%%%%%%%%%%%%%%%%%%%%%
\subsubsection{Estimating von Neumann entropy}
Next, we proceed to look at quantum algorithms for mixed quantum states. As far as we are aware, this work is the first to investigate multiplicative approximations of the von Neumann entropy of density matrices. We prove the following theorem showing that such approximation is possible with query complexity that is \emph{sublinear} in the dimension of the state.
%sublinear in the dimension of the quantum system, and illustrating the power of the purified quantum access input model that we work with. 

\begin{theorem}
\label{thm:quant}
    There is a quantum algorithm that outputs a $\gamma$-multiplicative approximation of the von Neumann entropy $S(\rho)$ of a density matrix $\rho\in\C^{n\times n}$ accessed via a purified quantum query oracle $U_{\rho}$ as in \cref{eq:purified-access-quant} using $U_{\rho}$ and its inverse $\widetilde{\O}\left( n^{\nicefrac12+\nicefrac{(1+\eta)}{2\gamma^2}}\right)$ times, provided $S(\rho)>3\gamma+\nicefrac{4}{\eta}$. 
\end{theorem}
To the best of our knowledge, this is the first example of an algorithm for estimating von Neumann entropy with sublinear complexity. In contrast, we remark that standard (tomographic) methods for learning the state, obtaining an additive estimate, or even \textit{testing} properties of its spectrum \cite{ODonnell2018Testing} typically require a number of samples or queries that scales linearly (for additive approximation) or even quadratically (for learning and testing) in the dimension $n$ of the system. We provide a detailed comparison of our algorithms with related works in \cref{sec:relatedwork}.

%%%%%%%%%%%%%%%%%%%%%%%%%%%
\subsubsection{Lower bounds}
We complement the foregoing upper bounds by proving lower bounds on the query complexity of quantum algorithms for multiplicative entropy estimation. For the general purified query access model (i.e., Model (iv)), in which we also show our upper bounds, we show that $\Omega\left(n^{\nicefrac{1}{3\gamma^2}}\right)$ uses of the quantum query oracle are necessary to $\gamma$-approximate the entropy of an unknown classical distribution, even when we are promised that the input has entropy larger than $\nicefrac{\log n}{\gamma^2}$. In fact, the aforementioned lower bound is proved via a reduction to a variant of the collision problem \cite{AaronsonShi04collisions} in the frequency vector model (i.e., Model (i)), and so it also holds for this stronger model. See details in \cref{sec:lb_freq}.

We also prove lower bounds in the quantum samples model (i.e., Model (iii)). This model is far stronger than the rest of the models, and in particular, it trivially admits $O(1)$-quantum-sample algorithms for problems such as uniformity testing, identity testing, and gap-support size testing, which are known to be hard in the other models. In the quantum samples model, we are able to prove a weak lower bound of $\Omega(\sqrt{\log n})$ by a reduction to the promise problem of testing identity of two known distributions in Hellinger distance \cite{Belovs2019}. To our knowledge, this constitutes the first non-trivial lower bound on the capability of this powerful input model. See details in \cref{sec:lb_qsample}.

%%%%%%%%%%%%%%%%%%%%%%%%%%%%%%%
\subsubsection{Additive approximation and gap problems}
\label{sec:applications}
Multiplicative approximation can generally also capture the notion of additive approximation, and in particular for entropies we can recover $\epsilon$-additive estimates by suitably choosing the multiplicative factor $\gamma$, while incurring only a small (logarithmic) overhead in the complexity. Multiplicative estimation has the added advantage of being closely related to the field of property testing, wherein we wish to test whether an input satisfies some global property (such as having high entropy) or is far from any possible input that has that property (say, in total variation distance). To illustrate the generality and utility of our results, we note the following immediate applications to additive estimation and testing.

\parahead{Estimating the entropy to additive precision: }Since both Shannon and von Neumann entropies of $n$-dimensional distributions or quantum systems are bounded by $\log n$, choosing  $\gamma=1+\frac{\epsilon}{\log n}$ we see that a good $\gamma$-multiplicative approximation also yields a good $\epsilon$-additive approximation (see \cref{sec:prelim-ests} for more details). Furthermore, the complexity overhead can be bounded by noting that
\begin{align}
    \frac{1}{2\gamma^2} &= \frac12\left(1+\frac{\epsilon}{\log n}\right)^{-2}\nonumber\\
        &\leq \frac12 + \frac{3\epsilon^2}{\log^2 n}.
\end{align}
Since the second term decays and is $o(1)$, we recover query complexities of $\O\left(n^{\frac12+o(1)}\right)$ and $\O\left(n^{1+o(1)}\right)$ for estimating the Shannon entropy of a probability distribution or the von Neumann entropy of a density matrix to constant additive precision. This matches the results of \cite{Gilyen2019Distributional}, upto to polylogarithmic factors (or equivalently, $n^{o(1)}$ factors). \\

\parahead{Testing whether the entropy is high or low: }
Suppose we wish to determine if a distribution on $[n]$ has entropy (1) larger than a threshold $H_1$, or (2) smaller than a threshold $H_2<H_1$ for some $H_1, H_2 \in(0,\log n]$. If we are able to $\gamma$-approximate the entropy with $\gamma=\sqrt{\frac{H_1}{H_2}}$, notice that in the first case the algorithm must output a value larger than $\sqrt{H_1 H_2}$, whereas in the second case it must output a value smaller than $\sqrt{H_1 H_2}$, hence allowing us to distinguish the two cases. Thus, we can solve this testing problem with nearly subquadratic quantum query complexity $\widetilde{\O}\left(n^{\frac{H_2}{2H_1}}\right)$. To compare, classical algorithms can solve this task with $\O\left(n^{\frac{H_2}{H_1}+o(1)}\right)$ samples and require $\Omega\left(n^{\frac{H_2}{H_1}-o(1)}\right)$ samples \cite{Valiant2011}.

%%%%%%%%%%%%%%%%%%%%%%%%%%%%%%%%%%%%%%%%%%%%%%%%%%%%%%%%%

\subsection{Related work}
\label{sec:relatedwork}
%%%%%%%%%%%%%%%%%%%%
%%%%%%%%%%%%%%%%%%%%%%%%%%%%%%%%%%%%%%%%%%%%%%%%%%%%%%%%
For easy reference, we collect in \cref{tab:class-comp,tab:quant-comp} the best known results on estimating Shannon and von Neumann entropies in input models of relevance to our work.

%%%%%%%%%%%%%%%%%%%%%%%%%%%%%%%%
\begin{table}[htb]
    \renewcommand{\arraystretch}{1.5}
    \centering
    \begin{adjustbox}{center}
    \begin{tabular}{|c|c|c|}
        \toprule
          \textbf{Type of estimate} & \textbf{Classical sample complexity} & \textbf{Quantum query complexity} \\ \midrule
          & {\tiny \cite{Jiao15additive,WuYang15additive}} & {\tiny \cite{Gilyen2019Distributional,Li2019RenyiQuery}~\&~\cite{Bun2018}}\\
         $\epsilon$-Additive  & $\Theta\left(\frac{n}{\epsilon \log n} + \frac{\log^2 n}{\epsilon^2}\right)$  & $\widetilde{\O}\left(\frac{\sqrt{n}}{\epsilon^{1.5}}\right)$ ~\&~ $\widetilde{\Omega}\left(\sqrt{n}\right)$  \\
         \specialrule{0.25pt}{5pt}{5pt}
         $\gamma$-Multiplicative & $\widetilde{\O}\left(n^{\frac{1+\eta}{\gamma^2}}\right)$~ \& ~ $\Omega\left(n^{\nicefrac{1}{\gamma^2}-o(1)}\right)$ & $\widetilde{\O}\left(n^{\frac{1+\eta}{2\gamma^2}}\right)$ \&~ $\Omega\left(n^{\nicefrac{1}{3\gamma^2}}\right)$\\
          & {\tiny \cite{Batu2002}$\qquad$\&$\qquad$\cite{Valiant2011}}& {\tiny (this work)}\\
         \bottomrule
    \end{tabular}
    \end{adjustbox}
    \caption{Classical and quantum sample and query complexities of estimating the Shannon entropy of classical distributions over an alphabet of size $n$, and $\eta>0$ controls the amount by which the entropy of the input is bounded away from zero (see \cref{thm:main} for details).
    }
    \label{tab:class-comp}
\end{table}
%%%%%%%%%%%%%%%%%%%%%%%%%%%%%%%%

%%%%%%%%%%%%%%%%%%%%%%%%%%%%%%%%
\begin{table}[htb]
    \renewcommand{\arraystretch}{1.5}
    \centering
    \begin{adjustbox}{center}
    \begin{tabular}{|l|c|L|c|}
        \toprule
           & \textbf{Type of estimate} & \textbf{Input model} & \textbf{Complexity} \\ \midrule
          
          \cite{Acharya2017QuantumAdditive} & $\epsilon$-Additive  & Copies of $\rho$ & $\O\left(\frac{n^2}{\epsilon^2}\right) ~\&~~ \Omega\left(\frac{n^2}{\epsilon}\right)$ \\
         \specialrule{0.25pt}{5pt}{5pt}
         \cite{Gilyen2019Distributional} & $\epsilon$-Additive & Purified quantum queries & $\widetilde{\O}\left(\frac{n}{\epsilon^{1.5}}\right)$  \\
         \specialrule{0.25pt}{5pt}{5pt}
         This work & $\gamma$-Multiplicative & Purified quantum queries & $\widetilde{\O}\left(n^{\nicefrac12+\nicefrac{(1+\eta)}{2\gamma^2}}\right)$ \& $\Omega\left(n^{\nicefrac{1}{3\gamma^2}}\right)$\\
         \bottomrule
    \end{tabular}
    \end{adjustbox}
    \caption{Comparing works on estimating the von Neumann entropy of an $n$-dimensional density matrix.
    }
    \label{tab:quant-comp}
\end{table}
%%%%%%%%%%%%%%%%%%%%%%%%%%%%%%%%

We can group studies of entropy estimation into four categories:
(1) classical and 
(2) quantum algorithms for estimating entropies of classical distributions;
(3) classical and (4) quantum algorithms for estimating the entropies of quantum states. 

We have already seen the most relevant works of the first kind in \cref{sec:intro}; it is worth remarking however that the estimation of entropies in a variety of classical input models and computational settings continues to be an active area of research.
    
We only note studies of the third category in passing:  \cite{Hastings2010MeasuringSimulations}, for instance, discuss a quantum Monte Carlo method to measure the $2$-R\'enyi entropy of a many-body system by evaluating the expectation value of a unitary swap operator.

At the intersection of categories (2) and (4), \cite{Acharya2017QuantumAdditive} study the sample complexity of estimating von Neumann and Renyi entropies of mixed states of quantum systems, in an input model where one gets access to $m$ independent copies of an unknown $n$-dimensional density matrix $\rho$. They allow arbitrary quantum measurements and classical post-processing, and show that in general the number of quantum samples required scales as $\Theta(n^2)$, which is asymptotically the same as the number of samples that would be required to learn the state completely via tomography. The experimental measurement of the entropy of certain kinds of quantum systems has also recently been studied, for example in \cite{Islam2015}. 
    
Other oracular input models may be potentially stronger than simple samples with measurement. The purified quantum query access model that we study in this work, wherein data is accessed in the form of a quantum state, is one such model. This state may be the output of some other quantum subroutine, in which case that subroutine itself is the oracle. Such input models can capture the fact that we have access to the process generating the unknown state, which we may \textit{a priori} expect to be useful in reducing the effort required in estimating its properties.
    
With regard to quantum algorithms for estimating the entropies of quantum states (which subsumes the case of classical probability distributions), \cite{Li2019RenyiQuery} provide upper and lower bounds on the query complexity for the task of additive approximation of von Neumann and Renyi entropies in the quantum frequency vector input model (see \cref{eq:freq-vec}). \cite{Gilyen2019Distributional} study another similar oracular model, known as the quantum purified query access model, which essentially provides a pure state, sampling from which reproduces the statistics of the original mixed state, or target classical distribution (see \cref{eq:purified-access-class,eq:purified-access-quant}). \cite{PhysRevA.104.022428} consider the estimation of Renyi entropies in the same purified query access model, to both additive and multiplicative precision. Their focus however is on how this task may be solved on restricted models of quantum computation (namely, $\mathsf{DQC1}$), and they do not obtain optimal query complexities. 
    
Finally, we remark that we use the approximation of the logarithm by power functions that is defined and studied in \cite{Zhao07}, who show that the Shannon entropy can be estimated to any desired precision by interpolation using estimates of R\'{e}nyi-$\alpha$ entropies for values of $\alpha\in(0,2]$. They study a streaming input model and use techniques that are otherwise very different from ours.

%%%%%%%%%%%%%%%%%%%%

\section{Preliminaries and notation}
We assume the reader is familiar with the quantum computing framework and notation, such as Dirac's bra-ket notation. We refer to standard texts such as \cite{Nielsen2010,deWolf2019notes} for a detailed introduction to quantum computation. Here, we discuss concepts and notation of specific relevance to this paper.

For $n\in\N$ we denote by $[n]$ the set $\{1,\ldots,n\}$. All logarithms that we use throughout this paper are taken with base $2$. For a probability distribution $\pbold=(p_1,\ldots,p_n)$ on $[n]$, we use the notation $\ww_{\pbold}(A):=\sum_{i\in A}p_i$ for the weight of a subset of labels $A\subseteq [n]$.

The Shannon entropy $H$ of $\pbold$ is defined by \cite{shannon}
\begin{equation}
    \label{eq:shannon_entropy}
    H(\pbold) \defeq -\sum_{i\in [n]} p_i \log p_i.
\end{equation}
We will let $H_{\pbold}(A):=-\sum_{i\in A}p_i\log p_i$ denote the entropy of the set of labels $A\subseteq[n]$ under the distribution $\pbold$. 

Quantum registers can also exist in probabilistic mixtures of states; the simpler superposition states are called pure states, and their probabilistic mixtures are known as mixed states. The most general description of an $n$-dimensional quantum state $\rho$ is in terms of an $n\times n$ positive semi-definite matrix with complex entries, normalised to have unit trace. The von Neumann entropy of a quantum state represented by its density matrix $\rho\in\C^{n\times n}$ is defined by \cite{MllerLennert2013}
\begin{equation}
    \label{eq:vonNeumann_entropy}
    S(\rho)\defeq-\Tr{\rho\log\rho}.
\end{equation}
For two distributions $\pbold$ and $\tilde{\pbold}$ on $[n]$, we define the Hellinger distance between them by
\begin{equation}
    \label{def:Hellinger}
    d_H(\pbold,\tilde{\pbold})\defeq \sqrt{\frac12\sum_i (\sqrt{p_i}-\sqrt{\tilde{p}_i})^2},
\end{equation}
and the total variation distance by
\begin{equation}
    \label{def:TV_distance}
    d_{\mathrm{TV}}(\pbold,\tilde{\pbold})\defeq \frac12\sum_i\Abs{ p_i-\tilde{p}_i}.
\end{equation}
We use the standard complexity theoretic notation of $\widetilde{\O}$ to hide polylogarithmic factors. Finally, all our algorithms succeed with constant probability which can be boosted by standard techniques, and we omit the resulting factors from the complexity and discussion for brevity.
%%%%%%%%%%%%%%%%%
\subsection{Multiplicative and additive estimates}
\label{sec:prelim-ests}
A good estimate $\tilde{x}$ of some unknown quantity $x>0$ to a multiplicative factor $\gamma>1$ satisfies
\begin{equation}
    \label{eq:multi_estimate}
    \frac{x}{\gamma}\leq \tilde{x}\leq \gamma x.
\end{equation} 
Similarly, an estimate $\tilde{x}$ of an unknown quantity $x$ to additive precision $\epsilon>0$ must satisfy
\begin{equation}
    \label{eq:add-estimate}
    x-\epsilon\leq \tilde{x}\leq x+\epsilon.
\end{equation} 
When we know an upper bound $0<X<\infty$ on $x$, we can obtain an $\epsilon$-additive approximation from a $\gamma$-multiplicative approximation by choosing $\gamma=1+\frac{\epsilon}{X}$, since 
\[
    x-\epsilon < \left(1-\frac{\epsilon}{X}\right) x < \left(1+\frac{\epsilon}{X}\right)^{-1} x < \tilde{x} < \left(1+\frac{\epsilon}{X}\right) x < x + \epsilon.
\]
In particular, we know that the entropy of any distribution over $[n]$ is bounded above by $\log n$, and so by choosing $\gamma=1+\frac{\epsilon}{\log n}$ we can always obtain good additive estimates using multiplicative estimation subroutines.

%%%%%%%%%%%%%%%%%%%%%%%%%
\subsection{Input models}
\label{sec:input-models}
We now formally define the four input models touched upon in \cref{sec:intro}. We refer to \cite{Belovs2019} for a more detailed discussion of these models and their relations to each other.
        \begin{enumerate}
            \item[(i)] \textbf{Frequency vectors with quantum query access: }A standard unitary quantum query oracle $U$ to a string $\mathbf{v}\in[n]^{m}$ for some large $m$, where $\forall~j\in[m]$, $i\in[n]$, and
            \begin{align}
                \label{eq:freq-vec}
                U\ket{i}\ket{0} &= \ket{i}\ket{v_i}\nonumber\\
                p_i &= \frac1n\Abs{\{j|v_j=i\}}.
            \end{align}
            
            \item[(ii)] \textbf{Quantum oracle that generates a sample from $\pbold$: }A standard unitary quantum query oracle to a string $\mathbf{v}\in[n]^{m}$ for some large $m$, where each $v_i$ is drawn independently at random according to the distribution $\pbold$.
            
            \item[(iii)] \textbf{Quantum samples for classical distributions}: A unitary that prepares the state
            \begin{align}
                U_{\pbold}\ket{0} = \sum_{i=0}^{d-1} \sqrt{p_i}\ket{i}
                    \defeq \ket{\pbold},
            \end{align}
            such that measuring the state $\ket\pbold$ in the computational basis reproduces the effect of sampling from $\pbold$.
            
            \item[(iv)] \textbf{Purified quantum query access}: A unitary $U_\rho$ on $\C^{n}\otimes\C^{n}$ which produces a purification $\ket{\uppsi_\rho}$ of the actual input state $\rho$ in $\C^{n\times n}$
            \begin{align}
            \label{eq:purified-access-quant}
                U_\rho\ket{0}  = \sum_{i=1}^n\sqrt{p_i}\ket{\psi_i}\ket{\phi_i}\defeq\ket{\uppsi_\rho}
            \end{align}
            such that the partial trace over the ancillary register $\Tr[2]{\ketbra{\uppsi_\rho}{\uppsi_\rho}}=\rho$. The states $\{\ket{\psi}\}$ and $\{\ket{\phi}\}$ are sets of orthonormal vectors on the system and ancillary subspaces respectively. Classical probability distributions can be considered as density matrices that are diagonal in the computational basis, so we consider a unitary $U_{\pbold}$ with a simplified action
            \begin{align}
            \label{eq:purified-access-class}
                U_{\pbold}\ket{0}  = \sum_{i=1}^n\sqrt{p_i}\ket{\psi_i}\ket{i}\defeq\ket{\uppsi_{\pbold}}.
            \end{align}
        \end{enumerate}
In all four cases, as is standard in the theory of quantum query complexity, we assume access to both the oracle and its conjugate $U^{\dagger}$.

The frequency vector model can emulate the purified access model for classical distributions, i.e., given a frequency vector oracle, we can construct a purified access oracle with a single query:
\begin{align}
    U\left(\frac{1}{\sqrt{m}}\sum_{j=1}^m\ket{j}\ket{0}\right) &= \frac{1}{\sqrt{m}}\sum_{j=1}^m\ket{j}\ket{v_j}\nonumber\\
        &= \sum_{i=1}^n\left(\frac{\displaystyle\sum_{j:v_j=i}\ket j}{\sqrt{m}}\right)\ket{i}\nonumber\\
        &= \sum_{i=1}^n\sqrt{p_i}\ket{\psi_i}\ket{i},
\end{align}
using the definition of the frequency vector and defining the normalised version of the state in parenthesis on the second line. A similar calculation shows model (ii) can also emulate model (iv). 

The vanilla quantum samples model for access to classical distributions too can emulate the purified query access model: applying a single layer of $\log n$ two-qubit $\CNOT$ gates to $\ket{\pbold}$ suffices to copy the computational basis states $\ket{i}$ into an ancillary register, reproducing the action in \eqref{eq:purified-access-class} for classical distributions. 

It is also worth noting that analogous to model (iii), for the case of mixed quantum states we may have access to multiple independent copies of the unknown state $\rho$, which is the model studied in works including \cite{Acharya2017QuantumAdditive,ODonnell2018Testing}.

%%%%%%%%%%%%%%%%%%%%%%%%%%%%%%%%%%%
\subsection{Algorithmic tools}
The main standard quantum algorithmic techniques that we use are the method of quantum singular value transformations (QSVT), quantum singular value estimation (QSVE), and quantum amplitude estimation (QAE). We give a brief and high level overview of these methods here.

\parahead{QSVT: }The QSVT \cite{Gilyen2018QuantumArithmetics} is a powerful framework for describing and constructing quantum algorithms. At its heart lies the linear algebraic formulation of quantum algorithms, and the theory of quantum walks. Given a unitary $U$ in $a+s$ dimensions that in a certain way encodes a possibly rectangular matrix $P$ with singular value decomposition $\sum_{i}\sigma_i\ketbra{v_i}{w_i}$, the map QSVT(P,f) uses $U$ as a black box and implements a unitary quantum circuit that approximately encodes a matrix $P_f$ with singular values $f(\sigma_i)$ transformed according to a polynomial $f$ defined on the values $\sigma_i$. This framework has been found to be immensely general in scope, in that most known quantum algorithms can be recast in terms of a QSVT. We will use this technique explicitly for the estimating the contribution to the entropy from points with high probability mass, in \cref{sec:heavywt}. We discuss more details regarding the QSVT of relevance to our work in \cref{app:pow-funcs}.

\parahead{QSVE: }The QSVE \cite{Kerenidis2017,Chakraborty2018,Gilyen2018QuantumArithmetics} is a generalisation of the popular and fundamental quantum phase estimation algorithm, generalising it to the task of coherently (i.e., in superposition) estimating the singular values of rectangular matrices. We use QSVE as a subroutine in \cref{sec:lightwt,sec:heavywt} for coherently distinguishing points with high probability mass from those with low probability mass. The subroutine QSVE(P,m) uses $U$ as a black box and maps an input state $\sum_i \alpha_i\ket{v_i}\ket 0$ to $\sum_i \alpha_i\ket{v_i}\ket{\tilde{\sigma}_i}$, where the $\tilde{\sigma}_i$ approximate the singular values $\sigma_i$ of $P$ to $m$ bits of precision. We give more details and a discussion on the relation between QPE and QSVE in \cref{app:phase-est}.

\parahead{QAE: }Being a subroutine that grew out of Grover's search algorithm and the associated amplitude amplification technique, QAE \cite{Brassard2002} estimates the amplitude a quantum state $\ket\psi=U\ket0$ puts on a particular flagged subspace, essentially by running QPE on a Grover iterate (or diffusion operator) constructed from the unitary $U$. The map QAE(flag, $\epsilon$) takes $U$ as input and outputs an $\epsilon$-additive estimate $\tilde{p}$ where $\ket\psi=\sqrt{p}\ket{\text{flag}}+\sqrt{1-p}\ket{\text{junk}}$. We use this technique to estimate various quantities obtained after processing by QSVT and QSVE in both \cref{alg:light,alg:Heavy}. For convenience we recall the standard formal statement of how QAE works in \cref{app:amp-est}.

%%%%%%%%%%%%%%%%%%%%%%%%%%%%%%%%%%%%%%%%%%%%%%%%%%%%%%%
\section{Estimating entropy to multiplicative precision using purified quantum queries}

In this section, we prove \cref{thm:class,thm:quant} by constructing our quantum algorithms for estimating Shannon and von Neumann entropies, and analysing their correctness and purified-quantum-query complexity. In fact, we prove the following, more general lemma, which can, with simple modifications, handle multiplicative approximation of any entropic functional $f:\mathcal{D}\to\R$ defined on vectors in $\mathcal{D}:=\R^{n}_{\geq 0}$, and may be of more general interest in distributional property testing, as well as in the context of space-bounded computation and streaming input models.

To formally state the lemma, we shall need the following notion of projected unitary encodings (which we will loosely and interchangeably refer to as block encodings), as defined in \cite{Gilyen2018QuantumArithmetics}.

\begin{definition}
\label{defn:block-encoding}
    An $(\alpha, a, \delta)$ projected unitary encoding of an operator $A$ acting on $s$ qubits is a unitary $U$ acting on $a+s$ qubits, such that 
    \begin{equation}
        \norm{A - \alpha\Pi^\dagger U\Pi} \leq \epsilon,
    \end{equation}
    where the first register consists of ancillary qubits,  $\Pi:=\ket{0}^{\otimes a}\otimes\id_s$ is an isometry mapping $(\C^2)^{\otimes s}\mapsto\mathrm{span}_{\C}\{\ket{0}^{\otimes a}\}\otimes (\C^2)^{\otimes s}$, and $\alpha,\epsilon\in(0,\infty)$.
\end{definition}

We will prove the following lemma for projected unitary encodings. Subsequently, using the techniques of \cite{Gilyen2019Distributional} to obtain projected unitary encodings from purified access oracles corresponding to classical distributions and density matrices respectively (see \cref{app:bloc-enc}), we will obtain \cref{thm:class,thm:quant} as immediate corollaries.

\begin{lemma}
    \label{thm:main}
    For any $\gamma>1$ and $0<\epsilon<1$, given an $(\alpha,a,\delta)$ projected unitary encoding $U$ of a matrix $P$ with singular values $\sqrt{p_1},\ldots,\sqrt{p_n}$ where $\pbold=(p_1,\ldots,p_n)$ defines a probability distribution, there is a quantum algorithm which outputs with high probability a $(1+2\epsilon)\gamma$-multiplicative estimate $\tilde{H}$ of the entropy of $H(\pbold)$, for distributions with entropy at least $3\gamma+\nicefrac{1}{2\epsilon}$.
    This algorithm makes $$m=\widetilde{\O}\left(\frac{\alpha n^{\nicefrac{1}{2\gamma^2}}}{\epsilon}\right)$$ uses of $U$ and $U^{\dagger}$, $\O(1)$ uses of controlled-$U$, and needs $\O(ma)$ additional one- and two-qubit gates.
\end{lemma}
\parahead{Remark 1.}
The statement of this result is in direct analogy with \cite[Theorem 1]{Batu2002}. In particular, in order for the algorithm to be correct we need the entropy of the input to be bounded away from zero, since no algorithm using any of the input models discussed in \cref{sec:input-models} can output multiplicative estimates of arbitrary distributions, as we will see in \cref{sec:zero-ent-lowerbnd}. If we desire a multiplicative factor of $\gamma$, we can first choose $\gamma'=\frac{\gamma}{(1+2\epsilon)}$. This leads to a small overhead in the complexity scaling as $\O\left(n^{\nicefrac{8\epsilon}{\gamma^2}}\right)$. Given $\eta>0$ we can rephrase this as saying that by choosing $\epsilon<\nicefrac{\eta}{8}$, the algorithm can deal with any distribution with $H(\pbold)=\Omega(\nicefrac{1}{\eta})$ using $\widetilde{\O}\left(n^{\frac{1+\eta}{\gamma^2}}\right)$ queries --- that is, we can weaken the promise on the input and enlarge the class of distributions that the algorithm is correct on by paying a small appropriate price in the complexity.\\
\\
\parahead{Remark 2.}
The $\widetilde{\O}$ in \cref{thm:main} hides factors that scale as (a) $\O(\log^2 n)$; (b) $\log\nicefrac{n}{\epsilon}$; and (c) $\O\left(\nicefrac{\sqrt{\gamma^3}}{\log\gamma}\right)$. It turns out that with a purified quantum query access oracle, we can always create an exact encoding $U$ of the distribution or quantum state with $\delta=0$, and so we will avoid discussing the dependence of the complexity on $\delta$ to avoid clutter.
\\

In particular, \cref{thm:main} captures both the case of Shannon entropy of classical probability distributions and von Neumann entropy density matrices, since the latter is definitionally the Shannon entropy of the eigenvalue spectrum of the density matrix (see \cref{sec:putting-together} for more details). We devote the rest of this section to proving \cref{thm:main}.

The theory behind our estimator is drawn from \cite{Batu2002}. We show how the estimator used therein can be computed more efficiently on a quantum computer with purified query access to the input. Recall that the Shannon entropy is defined by \cref{eq:shannon_entropy}
\[
    H(\pbold) = \sum_i p_i\log \frac{1}{p_i}.
\]
To estimate $H(\pbold)$, we first divide the domain into two sets, of `big' and `small' elements with respect to a choice of threshold $\beta\in(0,1)$:
\begin{equation*}
    B:=B_\beta=\{i\in [n] : p_i\geq \beta\}\;,
\end{equation*}
and 
\begin{equation*}
    S:=S_{\beta}=[n]\setminus B_{\beta}\;.
\end{equation*}
Then $H(\pbold)=H_{\pbold}(B)+H_{\pbold}(S)$, since the Shannon entropy is linear as a function of subsets of its domain. We will aim to make the threshold value $\beta$ as large as possible, and in particular, we would like it to scale inverse sublinearly as a function of $n$.

We start by considering the lightweight elements first.

%%%%%%%%%%%%%%%%%%%%%%%%%%%%%%%%%%%%%%%%%%%%%%%%%%
\subsection{Estimating the entropy of the low weight elements}
\label{sec:lightwt}

The set $S_{\beta}$ of light elements can contribute heavily to the entropy as evidenced for instance by the uniform distribution. However, the low probability mass of these elements can be hard to estimate. If the (unknown) weight of these elements is some $\ww_{\pbold}(S)$, Lemma 4 of \cite{Batu2002} gives us a way to handle $S_{\beta}$. 
\begin{lemma}
\label{lem:batu-lightwt}
    $\ww_{\pbold}(S)\cdot\log\nicefrac{1}{\beta}\leq H_{\pbold}(S) \leq \ww_{\pbold}(S)\cdot\log n + \frac{1}{\ee}$.
\end{lemma}
\begin{proof}
    The lower bound is attained by a distribution that has as many points as possible with extremal weight (equal to $\beta$ or $0$), e.g., by a distribution with $\frac{1}{\beta}\cdot \ww_{\pbold}(S)$ elements having probability mass $\beta$ and the rest being zero.
    
    The upper bound is given by the distribution that puts the weight $\ww_{\pbold}(S)$ uniformly on all of its $n$ points.
\end{proof}
This enables us to simply estimate the total weight of the light elements, and use this to improve our approximation to $H(\pbold)$, as we shall see below in \cref{sec:combine-B-S}. For the moment, we note that for the choice of threshold $\beta=n^{-\nicefrac{1}{\gamma^2}}$, \cref{lem:batu-lightwt} tightly bounds the entropy of the lightweight elements as lying between their net weight times $\frac{\log n}{\gamma^2}$ and the weight times $\log n$ (up to a small constant). 

Furthermore, we can also see from the upper bound in \cref{lem:batu-lightwt} that unless the net weight of the light elements scales at least inverse logarithmically in $n$, the contribution of the light set $S_\beta$ to the entropy is bounded by a constant. In particular, 
\begin{equation}
    \label{eq:eps_1}
    \ww_{\pbold}(S)=o\left(\frac{1}{\log n}\right) \implies H_{\pbold}(S) = \frac{1}{e} + o(1).
\end{equation}

Thus we may choose to estimate $\ww_{\pbold}(S)$ to additive precision $\epsilon_1=\O\left(\nicefrac{1}{\log^2 n}\right)$. Given an $(\alpha, a, \delta)$-projected unitary encoding of $P$, intuition suggests that we can perform Quantum Phase Estimation (QPE) with this encoding to single out the lightweight elements and then estimate their net amplitude by performing Quantum Amplitude Estimation (QAE, \cref{app:amp-est}). This indeed turns out to work. \\

%%%%%%%%%%%%%%%%%%%%%%%%%%%%%%%%%%%%%%%%%%%%%
\begin{algorithm}[ht]
\setstretch{2.2}
\caption{\textsc{LightWeight}$(\beta)$ -- Estimate $\ww_{\pbold}(S)$ to additive precision $\epsilon_1$.}
\label{alg:light}
\begin{algorithmic}[1]
    \State $m=\log\frac{1}{\sqrt{\beta}}$ \Comment{Number of bits of precision for QSVE}
    \State input $\gets U_{\pbold}\ket{0}\ket{0^{m}} = \displaystyle \sum_{i\in[n]}\sqrt{p_i}\ket{\phi_i}\ket{i}\ket{0^{m}}$  \Comment{Input state for QSVE}
        \State $\ldots \xrightarrow{~~\mathrm{QSVE}\left(P,~ m\right)~~} \displaystyle\sum_{i\in[n]}\sqrt{p_i}\ket{\phi_i}\ket{i}\ket{q_i}$ \Comment{$q_i=0 \iff p_i<\beta$}% \Comment{Apply QPE on the unitary block encoding of $\ee^{-2\pi\ii P}$ with $m$ bits of precision} 
        \Statex $\hfill= \displaystyle\sum_{i\in S}\sqrt{p_i}\ket{\phi_i}\ket{i}\ket{0^m}_{\text{flag}} + \displaystyle\sum_{i\in B}\sqrt{p_i}\ket{\phi_i}\ket{i}\ket{\perp}_{\text{flag}}$ \Comment{$\braket{\perp|0^m}=0$}
        \State Pick $\epsilon_1 = \nicefrac{\epsilon}{\log^2 n}$ and let \Comment{$\implies\mathrm{QAE}$ cost $=\O\left(\log^2 n\right)$}
        \Statex $\tilde{w}_{\pbold}(S)=\mathrm{QAE}(\text{flag}=\ket{0^m},\epsilon_1)$ \Comment{$\Abs{\tilde{w}_{\pbold}(S) - w_{\pbold}(S)}\leq\epsilon_1$}
         \\
    \Return $\tilde{w}_{\pbold}(S)$
\end{algorithmic}
\end{algorithm}
%%%%%%%%%%%%%%%%%%%%%%%%%%%%%%%%%%%%%%%%%%%%%%%%%%%%%%%%%%%%%

\parahead{Using quantum singular value estimation to flag the light subspace: }
We would like to use QPE as a subroutine to separately flag the subspaces of heavy and light elements. In essence, we want to perform the map
\begin{equation}
    \label{eq:qpe-qsve}
    \sum_{i\in[n]} \sqrt{p_i}\ket{\phi_i}\ket{i}\otimes \ket{0^{m}} \mapsto \sum_{i\in[n]} \sqrt{p_i}\ket{\phi_i}\ket{i} \ket{q_i},
\end{equation}
where $\Abs{\sqrt{p_i}-q_i}\leq 2^{-(m+1)} =: \epsilon$, and $m$ is the number of bits of precision. When given a unitary block encoding for the matrix $P$, this problem is termed Quantum Singular Value Estimation (QSVE) and is solved in \cite{Kerenidis2017,Chakraborty2018,Gilyen2018QuantumArithmetics}; the complexity of their algorithm is essentially $\widetilde{O}\left(\nicefrac{1}{\epsilon}\right)$ where $\epsilon$ is the precision to which we would like to estimate the singular values. We defer the full discussion of the intuition and details behind doing this to \cref{app:phase-est}. 

The accuracy of the QPE subroutine is normally defined in terms of the number of bits of precision. Thus, to obtain a clean split between heavy and light elements according to the chosen threshold 
\[
    \beta=n^{-\frac{1}{\gamma^2}} = 2^{-\frac{\log n}{\gamma^2}},
\]
recalling that $P$ has singular values $\sqrt{p_i}$, it will be convenient for us to first round $\sqrt{\beta}$ down to the nearest power of $2$, and then scale down the approximation factor $\gamma$ appropriately. To this end, we set
\[
    \sqrt{\beta'} = 2^{-\left\lceil\frac{\log n}{2\gamma^2}\right\rceil} =: n^{-\frac{1}{2\gamma'^2}},
\]
which suggests that we should choose a tighter approximation factor, given by
\[
    \gamma' = \gamma\cdot\sqrt{\frac{\nicefrac{\log n}{2\gamma^2}}{\left\lceil\nicefrac{\log n}{2\gamma^2}\right\rceil}}.
\]
Note that $\gamma'\leq \gamma$, and so any good $\gamma'$-approximation is also a good $\gamma$-approximation. With this choice, we can run QSVE with the block encoding of $P$ to $m=\log\frac{1}{\sqrt{\beta'}}$ bits of precision, implying an additive precision of $\epsilon=2^{-(m+1)}$, incurring a complexity of
\begin{align}
    \O\left(\frac{1}{\epsilon}\right) &= \O\left(\frac{1}{\sqrt{\beta'}}\right) \nonumber\\
        &=\O\left(n^{\frac{1}{2\gamma'^2}}\right)
        =\O\left(n^{\frac{1}{2\gamma^2}}\right),
\end{align}
where we have used that
\begin{align*}
    \frac{1}{\gamma^2}\cdot\frac{\left\lceil\nicefrac{\log n}{2\gamma^2}\right\rceil}{\frac{\log n}{2\gamma^2}} &\leq \frac{1}{\gamma^2}\cdot\frac{\frac{\log n}{2\gamma^2}+1}{\frac{\log n}{2\gamma^2}}\\
    & \leq \frac{1}{\gamma^2}\left(1+\frac{2\gamma^2}{\log n}\right),
\end{align*}
and 
$
    n^{\nicefrac{1}{\log n}} = \O(1).
$

\parahead{Correctness:} With these steps in hand, one can readily see that \cref{alg:light} outputs an additive estimate of the net weight of all the elements whose probability mass lies below the threshold $\beta$, i.e.,
\begin{equation}
    \label{eq:lightwt-est}
    \Abs{\tilde{\ww}_{\pbold}(S) - \ww_{\pbold}(S)} \leq \epsilon_1.
\end{equation}\\
\parahead{Complexity: }With the reasoning of \cref{eq:eps_1} and our choice of $\epsilon_1=\nicefrac{\epsilon}{\log^2 n}$, the query complexity of \cref{alg:light} in terms of queries to the block encoding of $P$ is given by
\begin{equation}
    % \label{eq:light-compl}
    \O\left(\frac{\log^2 n}{\epsilon}\right)\cdot\O\left(\frac{1}{\sqrt{\beta}}\right),
\end{equation}
where the first term is the complexity of the amplitude estimation step (\cref{app:amp-est}). For the choice $\beta = n^{-\nicefrac{1}{\gamma^2}}$, this becomes
\begin{equation}
    \label{eq:light-compl}
    \O\left(\frac{n^{\nicefrac{1}{2\gamma^2}}\log^2 n}{\epsilon}\right).
\end{equation}
Intuitively, as previously noted in \cite{Gilyen2019Distributional}, the projected unitary encoding of the input gives us \textit{operational access} to the square roots of the point probabilities. Hence, since $p_i<\sqrt{p_i}$, quantum algorithms can in a sense `see' lightweight elements more easily than classical algorithms.

\parahead{The normalisation $\alpha$ of the block encoding: } So far we have not discussed how the normalisation factor $\alpha\geq1$ and precision $\delta\in(0,1)$ of the $(\alpha, a, \delta)$ projected unitary encoding $U$ of $P$ affects the algorithm and its complexity. This is, fortunately, easy to do. First we note that the states the purified access oracles generate, both in \cref{eq:purified-access-class,eq:purified-access-quant}, actually encodes \textit{the exact} values of $\sqrt{p_i}$ and require no additional normalisation factor. Thus the only step where $\alpha$ and $\delta$ enter the picture are in the QSVE step. Next, as we remarked previously, the constructions that we use will have $\delta=0$, so we ignore this precision, which even otherwise would only contribute logarithmic factors to the complexity. Finally, since the singular values of $P$ are $\sqrt{p_i}$ and $U$ encodes $P/\alpha$, the normalisation factor effectively means that in the QSVE step we must work harder and estimate them to precision $\sqrt{\beta}/\alpha$ --- this directly contributes exactly a factor of $\alpha$ to the overall complexity in \cref{eq:light-compl} of \cref{alg:light}. 

This can also be seen directly from the complexity of QSVE in \cite[Theorem 27]{Chakraborty2018}, which scales as $\O(\nicefrac{\alpha}{\epsilon})$ for performing QSVE to precision $\epsilon$ using an $(\alpha,a,\delta)$ block encoding.

%%%%%%%%%%%%%%%%%%%%%%%%%%%%%%%%%%%%%%%%%%
\subsection{Estimating the entropy $H_{\pbold}(B)$ of the heavy elements}
\label{sec:heavywt}
So far we have looked at the set $S_{\beta}$ of lightweight elements. We now turn to estimating the contribution to the entropy from the heavy elements, i.e., those with probability mass greater than the threshold value of $\beta$. 

The first ingredient we use for this is a result from classical approximation theory that provides a multiplicative approximation of the logarithm by a combination of functions of the form $x^{\pm a}$ for small values of $a>0$ \cite{Zhao07}.

%%%%%%%%%%%%%%%%%%%%%
\subsubsection{Multiplicative approximation of entropy using power functions}
For any $a\in(0,1)$, consider the following functions 
\begin{equation}
    f_{\pm}(x) = x^{\pm a},~~~~~~
    f(x) = -\frac{f_{+}(x)-f_{-}(x)}{2a}. 
\end{equation}
On a domain $(\beta, 1]$ for some $\beta>0$, we have the following series expansions for $f_{\pm}(x)$ 
\begin{align}
    x^{\pm a} &= e^{\pm a\log x} \nonumber\\
            &= 1 \pm a\log x + \frac{(a\log x)^2}{2!} \pm\frac{(a\log x)^3}{3!} + \ldots,    
\end{align}
from which we deduce that
\begin{equation}
    f(x) = -\log x\cdot\left(1 + \frac{(a\log x)^2}{3!} +\frac{(a\log x)^4}{5!} + \ldots\right).
\end{equation}
Since $\log x<0$ on our interval of interest $x\in(\beta,1]$, it will be convenient to write $-\log x = \log\nicefrac{1}{x}>0$. We see that $f(x)$ is a one sided approximation for $\log x$, in the sense that it is always at least $\log \nicefrac{1}{x}$, and at most as large as $\log \nicefrac{1}{x}\cdot g(a,x)$, where the multiplicative factor is
\begin{align}
    g(a,x)&=1 + \frac{(a\log x)^2}{3!} +\frac{(a\log x)^4}{5!} + \ldots\nonumber\\ 
        % &\leq 1 + \frac{(a\log x)^2}{2!} +\frac{(a\log x)^4}{4!} + \ldots\nonumber\\ 
        &\leq 1 + |a\log x| + \frac{(a\log x)^2}{2!} + \frac{\Abs{a\log x}^3}{3!} +\frac{(a\log x)^4}{4!} + \ldots\nonumber\\
        &=\ee^{\Abs{a\log x}}.
\end{align}
Note that $\forall x\in (\beta, 1]$, $g(a,x) \leq g(a,\beta)$. Thus if we would like $f(x)$ to approximate $\log x$ to within a multiplicative factor $\gamma>1$ in the sense of \cref{eq:multi_estimate}, it suffices to choose $a \in(0,1)$ such that 
\begin{align*}
     g(a,\beta) \leq \ee^{a\log\nicefrac{1}{\beta}} &\leq \gamma\;,
\end{align*} 
which translates to requiring that
\begin{align}
    a \leq \frac{\log \gamma}{\log\nicefrac{1}{\beta}}.
\end{align}
Choosing such an $a$, we have a $\gamma$-multiplicative approximation of the logarithm for $x\in(\beta,1]$, i.e.,
\begin{equation}
    \label{eq:mult-approx-log}
    \frac{1}{\gamma}\cdot\log \nicefrac{1}{x} \le \log \nicefrac{1}{x}\leq f(x) \leq \gamma\cdot\log \nicefrac{1}{x}.
\end{equation}
It is worth remarking that this is actually \textit{stronger} than a $\gamma$-multiplicative approximation, by virtue of being one-sided; the lower bound on the approximation is actually as strong as $\log \nicefrac{1}{x}$, without the scaling by $\nicefrac{1}{\gamma}$. We will use this stronger inequality in our calculations below and freely replace it with the weaker version where required. 

For $i\in B_\beta$, if we replace the logarithm of $p_i$ with the function $f$, we get a multiplicative approximation of $H_{\pbold}(B)$. Indeed by \cref{eq:mult-approx-log}, $\forall p_i\in(\beta, 1]$
\begin{equation}
    \log \frac{1}{p_i} \leq f(p_i) \leq \gamma\cdot\log \frac{1}{p_i}.
\end{equation}
Multiplying by $p_i$ and summing over $i\in B_\beta$, we have
\begin{equation}
    \sum_{i\in B} p_i \log \frac{1}{p_i} \leq \sum_{i\in B} p_if(p_i) \leq \gamma\cdot\sum_{i\in B} p_i \log \frac{1}{p_i},
\end{equation}
and so we have 
\begin{equation}
    \label{eq:HB-mult-approx}
    H_{\pbold}(B) \leq \sum_{i\in B} p_i f(p_i) \leq \gamma\cdot H_{\pbold}(B)\;.
\end{equation}
We hence see that in fact for any $B\subseteq[n]$, $\displaystyle\sum_{i\in B} p_i f(p_i)$ is a good $\gamma$-multiplicative approximation to $H_{\pbold}(B)$. Next, we look at how to estimate the sum over $p_i f(p_i)$ by using the quantum singular value transformations (QSVT) technique \cite{Gilyen2018QuantumArithmetics} of implementing functions of block encoded matrices on quantum computers (see also \cite{Chakraborty2018,Subramanian_2019}).

%%%%%%%%%%%%%%%%%%%%%%%%%%
\subsubsection{Using QSVT to estimate power sums}
Defining the following power sums with exponent $a$, 
\begin{align}
    F_{\pm} = \sum_{i\in B} p_i^{1\pm a} = \sum_{i\in B} p_if_{\pm}(p_i),
\end{align}
we see that our estimator for $H_{\pbold}(B)$ in \cref{eq:HB-mult-approx} takes the form
\begin{equation}
    F = -\frac{F_{+}-F_{-}}{2a}.
\end{equation}
This suggests a strategy of estimating $F_{\pm}$ separately and combining them to obtain $F$.  Ideally, we would have liked to prepare two states, which have squared amplitude equal to $F_{\pm}$ on a subspace flagged by $\ket{0}$ in the ancilla. One possible form such states might take is
\[
    \sum_i \sqrt{p_i} f_{\pm}(\sqrt{p_i}) \ket{\psi_i}\ket{i}\ket{0} + \ket{\text{junk}}\ket{1}.
\]
Since we cannot directly implement arbitrary power functions $x^{\pm a}$, we use the standard technique of implementing the quantum singular value transformation corresponding to polynomial approximations $\tilde{f}_{\pm}(x)$ of $f_{\pm}(x)$ over the domain specified by $B$; recalling that we have a block encoding of $P$ whose singular value spectrum encodes $\sqrt{p_i}$, we only need to work with the domain $x\in[\sqrt{\beta},1]$. 

We discuss the details behind constructing and implementing the polynomials $\tilde{f}_{\pm}(x)$ in \cref{app:pow-funcs}. Intuitively, polynomial approximations using Taylor series give an exponential convergence in the degree of the approximating polynomial for smooth functions, i.e., the degree of the approximating polynomial only needs to grow as $\log\nicefrac{1}{\epsilon}$. Since the query complexity of QSVT depends on the degree of the polynomial being implemented, this in conjunction with what we noted above about the domain of approximation being $[\sqrt{\beta},1]$ leads us to expect the net query complexity of estimating $F_{\pm}$ to grow as $\O\left(\nicefrac{1}{\sqrt{\beta}}\right)$. We will show below that this is indeed the case.

%%%%%%%%%%%%%%%%%%%%%%%%%%%%%%%%%%%%%%%%%%%%%
\begin{algorithm}[ht]
\setstretch{2.2}
% \numberwithin{algorithm}{section}
% \renewcommand{\thealgorithm}{\arabic{section}.\arabic{algorithm}}
\caption{\textsc{HeavyEntropy}$(\beta)$ -- Estimate $H_{\pbold}(B)$ to multiplicative factor $\gamma>1$.}
\label{alg:Heavy}
\begin{algorithmic}[1]
    \State input $\gets U_{\pbold}\ket{0}\ket{0} = \displaystyle \sum_{i\in[n]}\sqrt{p_i}\ket{\phi_i}\ket{i}\ket{0}$  \Comment{Input state for QSVT}
        \State $\ldots \xrightarrow{~~\mathrm{QSVT}\left(P,~\tilde{f}_{\pm}\right)~~} \displaystyle\sum_{i\in [n]}\sqrt{p_i}\tilde{f}_{\pm}(\sqrt{p_i})\ket{\phi_i}\ket{i}\ket{0} + \ket{\text{junk}}\ket{1}$ \Comment{QSVT - \cref{app:pow-funcs}}
        \State $\ldots \xrightarrow{~~\mathrm{QSVE}\left(P,~ m\right)~~} \displaystyle\sum_{i\in B}\sqrt{p_i}\tilde{f}_{\pm}(\sqrt{p_i})\ket{\phi_i}\ket{i}\ket{0}\ket{0}_{\text{flag}} + \ket{\text{junk}}\ket{1}_{\text{flag}}$ % \Comment{Apply QPE on the unitary block encoding of $\ee^{-2\pi\ii P}$ with $m$ bits of precision} 
        \Statex  \Comment{$\text{flag}=0 \iff p_i\geq\beta$}
        \State Pick $\epsilon_3$ as in \cref{eq:heavy-eps3} and let 
        \Statex $\tilde{F}_{\pm}=\mathrm{QAE}(\text{flag}=\ket{0},\epsilon_3)$ \Comment{$\implies\mathrm{QAE}$ cost $=\O\left(\frac{1}{\epsilon_3}\right)$} \\
    \Return $\tilde{F} = -\frac{2\tilde{F}_{+}-2\sqrt{\gamma}\tilde{F}_{-}}{2a}$
\end{algorithmic}
\end{algorithm}
%%%%%%%%%%%%%%%%%%%%%%%%%%%%%%%%%%%%%%%%%%%%%

\parahead{Some difficulties:} In constructing polynomial approximations on our subdomain $[\sqrt{\beta},1]$ of interest, we need to admit a small interval $[\nicefrac{\sqrt{\beta}}{2},\sqrt{\beta}]$ where we allow the polynomial to vary before falling to low values on the rest of the domain $[-1,1]$. It is in accordance with this intuition that we have the guarantees $\Abs{f_{\pm}(x)}\leq 1$ on $[0,1]$, and $\Abs{f(x)}\leq \epsilon$ on $[0,\nicefrac{\sqrt{\beta}}{2})$ in \cref{app:pow-funcs}. If we na\"{i}vely try to use the simple state produced by applying the QSVT for $\tilde{f}_{\pm}$ as in step 2 of \cref{alg:Heavy} to estimate $F_{\pm}$, we see that the amplitude of the part of the state flagged by zero is actually given by
\begin{align}
    \sum_{i=1}^n p_i \tilde{f}^2_{\pm}(\sqrt{p_i}) &= \sum_{i\in B_\beta} p_i \tilde{f}^2_{\pm}(\sqrt{p_i}) + \sum_{i\in [\nicefrac{\beta}{2},\beta]} p_i  \tilde{f}^2_{\pm}(\sqrt{p_i}) \nonumber\\
        &\quad+\sum_{i\in [0,\nicefrac{\beta}{2}]} p_i  \tilde{f}^2_{\pm}(\sqrt{p_i}).
\end{align}
In particular, the second term on the rhs above is undesirable and in addition not easy to control. 
% : we have very little control over the value of the approximating polynomial in the interval $[\nicefrac{\beta}{2}, \beta]$, and so 
While we can upper bound the error caused by this term, we cannot manage it well without increasing the degree of the polynomial and hence incurring overheads in the query complexity. 

\parahead{Our approach:} It may in principle be possible to modify the approximating polynomial sufficiently with only polylogarithmic overheads in the degree, but we do not explore this route; instead we once again invoke QSVE as a tool to split the heavy and light weight subspaces. We may repeat the arguments of the previous section, with the slight modification of flagging the heavy elements with $\ket{0}$ in a single qubit ancilla, say by applying a unitary comparator circuit to check if the estimated singular value $q_i>\sqrt{\beta}$. 

\parahead{Correctness: }Thus, generating states of the form given in step 3 of \cref{alg:Heavy}, we see that their amplitudes in the flagged subspace now exactly square to
\begin{align}
    \tilde{F}_{\pm} = \sum_{i\in B} p_i \tilde{f}^2_{\pm}(\sqrt{p_i}),
\end{align}
and taking into account the normalisation factors corresponding to $\tilde{f}_{\pm}$, we define
\begin{equation}
    \tilde{F} = -\frac{2\tilde{F}_{+}-2\sqrt{\gamma}\tilde{F}_{-}}{2a} ,   
\end{equation}
where in particular, by the guarantees on the polynomial approximations \cref{app:pow-funcs} we know that $\forall x\in[\sqrt{\beta},1]$
\begin{align}
    \Abs{\tilde{f}_{-}(x) - \frac{\sqrt{\beta}^a}{2}x^{-a}} &\leq \epsilon_2\nonumber\\
    \Abs{\tilde{f}_{+}(x) - \frac{x^a}{2}} &\leq \epsilon_2.
\end{align}
Recalling that 
\[
    a = \frac{\log \gamma}{\log\nicefrac{1}{\beta}},
\]
we see that the normalisation factor for $f_{-}$ is given by
\[
    \beta^{\nicefrac{a}{2}} = \ee^{\log\beta\cdot\frac{\log\gamma}{-2\log\beta}} = \frac{1}{\sqrt\gamma}.
\]
From these approximation guarantees we immediately see that
\begin{align}
    \Abs{F_{\pm}-\tilde{F}_{\pm}} &\leq n\epsilon_2\nonumber\\
    \Abs{F-\tilde{F}} &\leq \frac{\sqrt{\gamma}n\epsilon_2}{a} % = \O\left(\sqrt{\gamma}n\epsilon_1 \cdot \frac{\log \nicefrac{1}{\beta}}{\log\gamma}\right)\;.
\end{align}
Furthermore, the QAE steps for $\tilde{F}_{\pm}$ are performed to some precision $\epsilon_3$, yielding 
$\hat{F}_{\pm}$ which satisfy
\begin{equation}
    \Abs{\hat{F}_{\pm} - \tilde{F}_{\pm}} \leq \epsilon_3,
\end{equation}
and so we also have the analogous quantity $\hat{F}$ such that
\begin{equation} 
    \Abs{\hat{F} - \tilde{F}} \leq \frac{2\sqrt{\gamma}\epsilon_3}{a} % =  \O\left(\sqrt{\gamma}\epsilon_2\cdot \frac{\log\nicefrac{1}{\beta}}{\log\gamma}\right).
\end{equation}
As indicated by the above expressions for the error, let us then choose the precision of the polynomial approximation and QAE steps such that the errors above are of constant order, i.e.,
\begin{align}
\label{eq:heavy-eps3}
    \epsilon_2 = \frac{a}{2n\sqrt\gamma}\cdot \epsilon &= \frac{\log \gamma}{2\sqrt\gamma n\log\nicefrac{1}{\beta}}\cdot \epsilon\nonumber\\
    \epsilon_3 = \frac{a}{4\sqrt\gamma}\cdot \epsilon &= \frac{\log \gamma}{4\sqrt\gamma \log\nicefrac{1}{\beta}}\cdot \epsilon.
\end{align}
This ensures that $\hat{F}$ is a good $\epsilon$-additive approximation to $F$, since
\begin{align}
    \Abs{F - \hat{F}} \leq \Abs{F - \tilde{F}} + \Abs{\tilde{F} - \hat{F}} \leq \epsilon.
\end{align}
Since this means $F - \epsilon \leq \hat{F} \leq F + \epsilon$, recalling that $F$ is a good $\gamma$-multiplicative approximation of $H_{\pbold}(B)$ as in \cref{eq:HB-mult-approx}, we arrive at the
% \begin{equation}
%     \frac{H_{\pbold}(B)}{\gamma} - \epsilon \leq H_{\pbold}(B) - \epsilon \leq &\hat{F} \leq \gamma H_{\pbold}(B) + \epsilon.
% \end{equation} 
%That is, we have the 
following multiplicative guarantee with some additive slack on our estimate $\hat{F}\equiv\tilde{H}_{\pbold}(B)$
\begin{equation}
    \label{eq:HB-est}
    \frac{H_{\pbold}(B)}{\gamma} - \epsilon \leq \hat{F} \leq \gamma H_{\pbold}(B)  + \epsilon.
\end{equation}

\parahead{Complexity: }The total query complexity of \cref{alg:Heavy} is the sum of the complexity of QSVT and QSVE corresponding to the polynomial approximations of $f(x)=x^{\pm a}$, multiplied by the complexity of QAE. The former two are respectively given by the degrees of $\tilde{f}_{\pm}$ from \cref{app:pow-funcs} and $\O(\nicefrac{1}{\sqrt{\beta}})$, both of which scale as $\widetilde{\O}(\nicefrac{1}{\beta})$, while the latter is the inverse of $\epsilon_3$. Thus, the net complexity is bounded by
\begin{align}
    \label{eq:heavy-compl}
    \left[\O\left(n^{\nicefrac{1}{2\gamma^2}}\right) + \O\left(n^{\nicefrac{1}{2\gamma^2}}\log\frac{n\log n}{\epsilon\log\gamma}\right)\right]&\cdot\O\left(\frac{\log n}{\epsilon\log\gamma}\right)\nonumber\\
        &\qquad\qquad=\widetilde{\O}\left(\frac{n^{\nicefrac{1}{2\gamma^2}}\log^2 n}{\epsilon\log\gamma}\right).
\end{align}

\parahead{The normalisation $\alpha$ of the block encoding: }First, we essentially repeat the discussion at the end of \cref{sec:lightwt}, giving a factor of $\alpha$ in the complexity of the QSVE step in \cref{alg:Heavy}.

In addition, we have to also account for the fact that the power functions are now implemented on the spectrum of $P/\alpha$ (for more details, see \cite[Theorem 56]{Gilyen2018QuantumArithmetics}). This has two effects: firstly, the domain over which we construct the polynomial approximations $\tilde{f}_{\pm}$ should now be $[\nicefrac{\sqrt{\beta}}{\alpha},1]$, which increases their degree by a factor of $\alpha$. Secondly, we obtain the power sums $F_+$ and $F_-$ are divided and multiplied respectively by an extra $\alpha^a$ factor, which means we need to choose $\epsilon_2$ and $\epsilon_3$ to be smaller by this factor. 

Since $\alpha=\frac{\log \gamma}{\log\nicefrac{1}{\beta}}$, for polynomially scaling normalisations $\alpha=n^{c}$ and $\beta=n^{-\nicefrac{1}{\gamma^2}}$ we see that $\alpha^a$ scales as $\gamma^{c\gamma^2}=\O(1)$ for our purposes. 

Hence since the complexities of the QSVT and QSVE steps add, we see once again that the net overhead in complexity is a factor of $\alpha$. 
%%%%%%%%%%%%%%%%%%%%%%%%%%%%%%%%%%%
\subsection{Combining the heavy and light estimates to form $\tilde{H}$}
\label{sec:combine-B-S}
Finally, we analyse the errors in and combine the estimates of $H_{\pbold}(B)$ and $H_{\pbold}(S)$ from \cref{alg:light,alg:Heavy} to get an estimate of $H(\pbold)$.
%%%%%%%%%%%%%%%%%%%%%%%%%%%%%%%%%%%%%%%%%%%%%
\begin{algorithm}[ht]
\setstretch{2}
% \numberwithin{algorithm}{section}
% \renewcommand{\thealgorithm}{\arabic{section}.\arabic{algorithm}}
\caption{Approximating $H({\pbold})$ to multiplicative precision $\gamma>1$}
\label{alg:main}
\begin{algorithmic}[1]
    \State $\gamma'\gets \gamma\cdot\sqrt{\frac{\nicefrac{\log n}{\gamma^2}}{\lceil\nicefrac{\log n}{\gamma^2}\rceil}}$ \Comment{Scale approximation factor down}
    \State $\beta \gets n^{-\nicefrac{1}{\gamma'^2}}$
    \Comment{Choose threshold that splits heavy \& light elements}
        \State $\tilde{H}_{\pbold}(B) \gets \Call{HeavyEntropy}{\beta}$ \Comment{\small $\frac{H_{\pbold}(B)}{\gamma} - \epsilon \leq \tilde{H}_{\pbold}(B) \leq \gamma H_{\pbold}(B)  + \epsilon$}
        \State $\tilde{w}_{\pbold}(S) \gets \Call{LightWeight}{\beta}$ \Comment{\small $\Abs{\tilde{\ww}_{\pbold}(S) - \ww_{\pbold}(S)} \leq \epsilon_1$}\\
    \Return $\tilde{H}_{\pbold}(B) + \frac{\tilde{w}_{\pbold}(S)\log n}{\gamma'}$
\end{algorithmic}
\end{algorithm}
%%%%%%%%%%%%%%%%%%%%%%%%%%%%%%%%%%%%%%%%%%%%%
Thus far, we have estimated the entropy of the heavy elements to multiplicative precision, viz \cref{eq:HB-est},
% \begin{equation*}
%     \frac{H_{\pbold}(B)}{\gamma} \leq \tilde{H}_{\pbold}(B) \leq \gamma H_{\pbold}(B),
% \end{equation*}
and we have estimated the total weight of the lightweight elements to additive precision \cref{eq:lightwt-est}.
% \begin{equation*}
%     \ww_{\pbold}(S) -\epsilon \leq \tilde{\ww}_{\pbold}(S) \leq \ww_{\pbold}(S)+\epsilon.
% \end{equation*}
The estimate that our algorithm outputs for $H(\pbold):= H_{\pbold}([n])$ is the quantity $\tilde{H}$ defined by
\begin{equation}
    \tilde{H} = \tilde{H}_{\pbold}(B) + \frac{\tilde{\ww}_{\pbold}(S)\log n}{\gamma}.
\end{equation}
\parahead{Upper bounding $\tilde{H}$:} Using the upper bounds on $\hat{F}$ and $\tilde{\ww}_{\pbold}(S)$, we see that
\begin{align}
    \tilde{H} \leq \gamma H_{\pbold}(B)  + \epsilon + \frac{\ww_{\pbold}(S) + \epsilon_1}{\gamma}\cdot \log n. 
\end{align}
The lower bound on $H_{\pbold}(S)$ in \cref{lem:batu-lightwt} along with the choice $\beta=n^{-\nicefrac{1}{\gamma^2}}$ implies that
\begin{equation}
    \frac{\ww_{\pbold}(S)\log n}{\gamma} \leq \gamma H_{\pbold}(S).
\end{equation}
Recalling that $\epsilon_1=\nicefrac{\epsilon}{\log^2 n}$, we thus have
\begin{align}
    \tilde{H} &\leq \gamma H_{\pbold}(B) + \gamma H_{\pbold}(S) + 2\epsilon \nonumber\\
        &\leq \gamma H(\pbold) + 2\epsilon\nonumber\\
        &\leq \left(1+2\epsilon\right)\gamma H(\pbold),
\end{align}
where on the last line we assume that $H(\pbold)\geq \frac{1}{\gamma}$, i.e., that the entropy of the input distribution is bounded away from zero, with a small increase in the approximation factor. This is a reasonable assumption to make, since no algorithm can output a good $\gamma$-multiplicative approximation of \textit{all} distributions (see \cref{sec:zero-ent-lowerbnd} for more details). 

\parahead{Lower bounding $\tilde{H}$:} This time using the lower bounds on $\hat{F}$ and $\tilde{\ww}_{\pbold}(S)$, we see that
\begin{align}
    \tilde{H} \geq H_{\pbold}(B)  - \epsilon + \frac{\ww_{\pbold}(S) - \epsilon_1}{\gamma}\cdot \log n. 
\end{align}
The upper bound on $H_{\pbold}(S)$ in \cref{lem:batu-lightwt} implies that
\begin{equation}
    \frac{\ww_{\pbold}(S)\log n}{\gamma} \geq \frac{H_{\pbold}(S)-\frac{1}{e}}{\gamma}.
\end{equation}
Thus we have that
\begin{align}
    \tilde{H} &\geq H_{\pbold}(B) + \frac{H_{\pbold}(S)}{\gamma} - 2\epsilon -\frac{1}{e\gamma} \nonumber\\
        &\geq \frac{H(\pbold)}{\gamma} - 2\epsilon-\frac{1}{\gamma}\nonumber\\
        &\geq \frac{H(\pbold)}{\left(1+2\epsilon\right)\gamma},
\end{align}
where to go from the second to the third line we assume that $H(\pbold)\geq 3\gamma + \nicefrac{1}{2\epsilon} \geq \frac{1}{\gamma}$, i.e., as before, that the entropy of the input distribution is bounded away from zero, with a small increase in the approximation factor. 

This shows that \cref{alg:main} outputs a $(1+2\epsilon)\gamma$-multiplicative approximation of $H(\pbold)$ and thus establishes its correctness. 

\parahead{Complexity: }Finally, the query complexity of \cref{alg:main} is the sum of the complexities of steps (3) and (4). From the complexities \cref{eq:light-compl,eq:heavy-compl} of \cref{alg:light,alg:Heavy} respectively, we have a net query complexity that scales as
\begin{equation}
    \label{eq:final-compl}
    \O\left(\frac{\alpha n^{\nicefrac{1}{2\gamma^2}}\log^2 n}{\epsilon}\right) + \widetilde{\O}\left(\frac{\alpha n^{\nicefrac{1}{2\gamma^2}}\log^2 n}{\epsilon\log\gamma}\right) = \widetilde{\O}\left(\frac{\alpha n^{\nicefrac{1}{2\gamma^2}}\log^2 n}{\epsilon\log\gamma}\right),
\end{equation}
which completes our proof of \cref{thm:main}.

%%%%%%%%%%%%%%%%%%%%%%%%%%%%%%%%%%%%%%%%%%%%%%%%
\subsection{Proving \cref{thm:class,thm:quant}}
\label{sec:putting-together}
For classical probability distributions accessed via a purified quantum query oracle $U_{\pbold}$ as in \cref{eq:purified-access-class}, we can construct projected unitary encodings of the kind required in \cref{thm:main} with $\alpha=1$ (see \cref{app:bloc-enc}). This immediately furnishes a proof of \cref{thm:class}. 

Similarly, for an arbitrary $n$-dimensional quantum density matrix accessed via a purified quantum query oracle $U_{\rho}$ as in \cref{eq:purified-access-quant}, we can construct a projected unitary encoding with $\alpha=\sqrt{n}$ (see \cref{app:bloc-enc}). Since $S(\rho)$ is equal to the Shannon entropy of the spectrum of $\rho$, we can plug this encoding into \cref{thm:main} to obtain a proof of \cref{thm:quant}. 

The key difference between the case of classical distributions and quantum mixed states is that for the former, we know that the purified access oracle produces a superposition over computational basis states in the second register, and we can use this knowledge to our advantage. On the other hand, for the latter case we do not \textit{a priori} know the basis in which the quantum state is diagonal, which reflects in the fact that we do not know the states $\ket{\psi_i}$ and $\ket{\phi_i}$ appearing in \cref{eq:purified-access-quant} beforehand.

%%%%%%%%%%%%%%%%%%%%%%%%%%%%%%%%%%%%%%%%%%%%%%%%%%%%%%%%%%%%%%%%%%%%%%%%%%%%%

\section{Lower bounds}
Batu et al.\ \cite{Batu2002} proved that even if we restrict to distributions with entropy $H(\pbold)\geq \nicefrac{\log n}{\gamma^2}$, any algorithm that estimates $H(\pbold)$ within a multiplicative $\gamma>1$ requires $\Omega(n^{\nicefrac{1}{2\gamma^2}})=\O(\sqrt{n})$ samples.  By arguing about the fingerprints of samples drawn from the unknown distribution, for small approximation factors $\gamma\in(1,\sqrt{2})$ they were able to show a stronger lower bound of $\Omega(n^{\nicefrac{2}{(5\gamma^2-2)}})$ samples, which is $o(n^{\nicefrac23})$, even when the input is known to be a distribution with $H(\pbold)\geq \frac{5\log n}{10\gamma^2-4}$. This was later improved by Valiant \cite{Valiant2011} to $\Omega(n^{\nicefrac{1}{\gamma^2}-o(1)})$, which showed the original upper bound of Batu et al.\ to be essentially tight.

The intuition behind proving lower bounds is to notice that estimating the entropy to a suitable multiplicative factor can suffice to distinguish between a given pair of distributions $\pbold$ and $\tilde\pbold$. Recall that any $\gamma$-approximation algorithm must output an estimate $\tilde{H}$ such that
\begin{align}
    \frac{H(\pbold)}{\gamma} &\leq \tilde{H}(\pbold) \leq \gamma H(\pbold) \nonumber\\
    \frac{H(\tilde\pbold)}{\gamma} &\leq \tilde{H}(\tilde\pbold) \leq \gamma H(\tilde\pbold).
\end{align}
If the ratio of entropies is larger than $\gamma^2$, then we have that
\begin{equation}
    \tilde{H}(\tilde\pbold) \leq \gamma H(\tilde\pbold) \leq \frac{H(\pbold)}{\gamma} \leq \tilde{H}(\pbold),
\end{equation}
and so $\gamma$-estimating $H$ will allow us to distinguish $\pbold$ and $\tilde\pbold$.

In this section we prove lower bounds similar to those of \cite{Batu2002} but for the case of quantum algorithms that output a $\gamma$-multiplicative estimate of the Shannon entropy of a classical distribution. We prove our bounds for the quantum frequency vector model. Recalling that this model is capable of emulating both the purified access model and the model that quantumly queries a classical list of samples, these lower bounds carry over to both these models as well. The technique we use for this is a reduction from the collision problem to multiplicatively approximating entropy.

We also prove what is perhaps the first non-trivial lower bound for the vanilla quantum samples model by showing a reduction from the promise problem of distinguishing two classical distributions in Hellinger distance to $\gamma$-multiplicative approximation of entropy. 

%%%%%%%%%%%
\subsection{Distributions with non-zero entropy}
\label{sec:zero-ent-lowerbnd}
We first show that no algorithm working with input models (i)-(iv) that makes only polynomially many queries can estimate the Shannon entropy of all distributions over $[n]$ to multiplicative precision.
Consider $\pbold=(1-\epsilon, \frac{\epsilon}{n-1},\ldots, \frac{\epsilon}{n-1})$ and $\tilde{\pbold}=(1,0,\ldots,0)$. The Hellinger distance between $\pbold$ and $\tilde{\pbold}$ is given by
\begin{align}
    d_H(\pbold,\tilde{\pbold})&= \sqrt{\frac12\sum_i (\sqrt{p_i}-\sqrt{\tilde{p}_i})^2}\nonumber\\
        &=\sqrt{1-\sqrt{1-\epsilon}}.
\end{align}
The binomial theorem tells us that for $\Abs{\epsilon}\le 1$ and $\beta\in\R$,
\begin{align}
    (1-\epsilon)^{\beta} &= 1+\sum_{k=1}^{\infty} \frac{(\beta)(\beta-1)\ldots(\beta-k)}{k!}(-\epsilon)^k.
\end{align} 
% from which we have that for $\beta=\nicefrac12$, 
% \begin{align}
%     (1-\epsilon)^{\nicefrac12} &< 1-\frac{\epsilon}{2}, % \nonumber\\
%     % (1-\epsilon)^{-\nicefrac12} &> 1+\frac{\epsilon}{2}.
% \end{align}
Therefore $1-\epsilon<(1-\epsilon)^{\nicefrac12}<1-\nicefrac{\epsilon}{2}$, and we have 
\begin{equation}
\sqrt{\epsilon} \geq d_H(\pbold,\tilde{\pbold}) \geq \sqrt{\frac{\epsilon}{2}},
\end{equation}
i.e., $d_H(\pbold,\tilde{\pbold})=\Theta(\sqrt{\epsilon})$. 

We also have $H(\tilde{\pbold})=0$ and
\begin{align}
    H(\pbold) &= -(1-\epsilon)\log (1-\epsilon) - \epsilon \log \epsilon + \epsilon\log (n-1)\nonumber \\
        % &= h(\epsilon) + \epsilon\log(n-1)\nonumber \\
        & = \Omega(\epsilon\log n),
\end{align}
since $h(\epsilon)\defeq -(1-\epsilon)\log(1-\epsilon) - \epsilon\log \epsilon \in [0,\log 2]$ is the binary entropy. Thus any algorithm that outputs an approximation for $H$ to a multiplicative factor $\gamma$ must output exactly $0$ on input $\tilde{\pbold}$ and at least $\frac{\epsilon}{\gamma}\log n$ on input $\pbold$. 

From \cite{Belovs2019}, we know that distinguishing $\pbold$ and $\tilde{\pbold}$ has query complexity 
\[
\Theta\left(\frac{1}{d_H(\pbold,\tilde{\pbold})}\right),
\]
in any of the four input models (i)-(iv). In particular, any algorithm requires $\Omega\left(\nicefrac{1}{\sqrt{\epsilon}}\right)$ queries to distinguish $\pbold$ from $\tilde{\pbold}$. Picking $\epsilon = n^{-k}$, for $\forall k>0$, shows that no $n^k$-query algorithm can distinguish $\pbold$ and $\tilde{\pbold}$, whence no such algorithm can output a good multiplicative approximation of the entropy for arbitrary input distributions.

%%%%%%%%%%%%%%%%%%%%%%%%%%%%%%%%%%%%%%%%
\subsection{General sub-logarithmic lower bounds}
\label{sec:lb_qsample}
Leaning on the lower bound in \cite{Belovs2019} for the promise problem of distinguishing two probability distributions, we also get a weak lower bound on the query complexity of entropy estimation for any input model, and in particular, for the vanilla quantum samples model.

Consider the distributions $\pbold = (1-\epsilon, \frac{\epsilon}{n-1}, \ldots, \frac{\epsilon}{n-1})$ and $\tilde{\pbold} = (1-\epsilon, \epsilon, 0, \ldots, 0)$ with
\begin{align}
    H(\pbold) &= \Omega\left(\epsilon \log n \right)\nonumber\\
    H(\tilde\pbold) &= h(\epsilon) \leq \log 2,
\end{align}
so that the ratio of entropies is
\begin{align}
    \frac{H(\pbold)}{H(\tilde\pbold)} \geq 1+\frac{\epsilon\log (n-1)}{\log 2} = \Omega(\epsilon \log n).
\end{align}
The Hellinger distance between these two distributions is given by
\begin{align}
    d_H(\pbold,\tilde{\pbold})&= \sqrt{\frac12 \left(\sqrt{\frac{\epsilon}{n-1}}-\sqrt{\epsilon}\right)^2 + \sum_{i=3}^n \frac{\epsilon}{n-1}}\nonumber\\
        &=\sqrt{\epsilon\left(1-\frac{1}{\sqrt{n-1}}\right)}\nonumber\\
        & \leq \sqrt{\epsilon},
\end{align}
so that the inverse of the Hellinger distance is of order $\Omega(\nicefrac{1}{\sqrt{\epsilon}})$. If we now make the choice
\begin{equation}
    \epsilon = \frac{\gamma^2}{\log n},
\end{equation}
Belovs' query lower bound for distinguishing $\pbold$ and $\tilde\pbold$ translates into a
\begin{equation}
    \Omega\left(\frac{\sqrt{\log n}}{\gamma}\right)
\end{equation}
lower bound for $\gamma$-estimating $H(\pbold)$. 

%%%%%%%%%%%%%%%%%%%%%%%%%%%%%%%%%%%%%%%%%%%%%%%%%%%%%%%%
\subsection{Polynomial lower bounds in the frequency vector and purified access models}
\label{sec:lb_freq}
For $\gamma>1$, consider the domain $[N]$ of size $N=n\cdot n^{1/\gamma^2}$, and consider the uniform distribution $\pbold$ on $[N]$, and a family of uniform distributions on subsets $S\subset[N]$ of size $|S|=n^{1/\gamma^2}$. 

For any such distribution $\tilde{\pbold}$, an input vector of length $N$ in the frequency vector input model represents an $r$-to-1 function $f:[N]\to S$, where $r=N/|S|=n$. On the other hand for inputs of this length $\pbold$ corresponds to a 1-to-1 function since each label in $[N]$ must occur exactly once in the input string.

Note then that the ratio of Shannon entropies is
\begin{equation}
    \frac{H(\pbold)}{H(\tilde{\pbold})} = \gamma^2 + 1 > \gamma^2,
\end{equation}
so that estimating $H$ to multiplicative precision $\gamma$ will enable us to distinguish the two distributions, and by extension, the two corresponding functions in the frequency vector input model. \cite{AaronsonShi04collisions} show that distinguishing a 1-to-1 function from an $r$-to-1 functions requires $\Omega\left((\nicefrac{N}{r})^{1/3}\right)$ queries to the input function oracle, where $N$ is the size of the domain of the functions. This for us translates to a lower bound of 
\[
       \Omega\left(\left(\frac{N}{r}\right)^{1/3}\right) = \Omega\left(n^{\nicefrac{1}{3\gamma^2}}\right).
\]

Recalling from \cref{sec:input-models} that any algorithm in the purified query access model implies an algorithm with the same complexity in the frequency vector model, and the fact that classical distributions are automatically examples of density matrices that are diagonal in the computational basis, we see that this lower bound applies to the estimation of both Shannon and von Neumann entropies, and to Models (i), (ii), and (iv).

%%%%%%%%%%%%%%%%%%%%%%%%%%%%%%%%%%%%%%%%%%%%%
\section{Conclusions and outlook}
In this paper, we initiated the investigation of quantum algorithms of sublinear query complexity for the task of $\gamma$-multiplicative approximation of both Shannon and von Neumann entropies. Our algorithm for probability distributions achieves a quadratic quantum speedup over classical algorithms, whilst our algorithm for mixed states indicates that it may be possible to estimate other global properties of quantum states with sublinear query complexity in their dimension.

Our results throw some light on the interesting question of the relation between the four input models discussed in \cref{sec:input-models}, which was first raised by \cite{Belovs2019}. In particular, the sub-logarithmic lower bound we obtain for the quantum samples model shows that there are problems that cannot be solved in this model with complexity independent of the dimension, where to our knowledge no such non-trivial problems were previously known for this model. It still remains open whether or not in the quantum samples model (which is strictly stronger than the general purified access models that we study), stronger speedups are possible for entropy estimation and similar tasks.

An immediate question left open by our work is to tighten the lower bounds we obtain, both for Shannon and von Neumann entropies. For the latter, the quantum polynomial method \cite{Bun2018} applied to the frequency vector model might yield better bounds. 
%while we suspect the latter requires new insight and reductions from the mixed state discrimination problem.
On the other hand, the strong intuition that quantum algorithms typically achieve quadratic speedups indicates that the upper bounds we obtain are tight up to polylogarithmic factors. A potential way to improve on these polylogarithmic factors may be to refine the approximation of $\log x$ by constructing functions such as $x^{2a}+x^{a}-x^{-a}-x^{-2a}$, reminiscent of symmetric Laurent series.

Our methods can also be extended to other information quantities such as Renyi and Tsallis entropies, and Kullback-Leibler and other divergence measures. Tight bounds on the complexity of multiplicative approximation of these quantities for both probability distributions and mixed states appear within reach, and we hope to report results in this direction in future work. 

We close by remarking that this line of work has close and interesting connections to distributional property testing, a rich and active field in classical complexity theory, offering exciting avenues for investigation in quantum complexity theory.

%%%%%%%%%%%%%%%%%%%%%%%%%%%%%%%%%%%%%%%%%%%%%%%%%%%%%%%%%%%%%%%%%%%%%%%%%%%%%%%%%%%%%%%%%%
% print biblatex bibliography
%%%%%%%%%%%%%%%%%%%%%%%%%%%%%%%%%%%%%%%%%%%%%%%%%%%%%%%%%%%%%%%%%%%%%%%%%%%%%%%%%%%%%%%%%%

\subsection*{Acknowledgements}
We thank Tugkan Batu and Cl\'ement Canonne for insightful discussions. 
%Tom Gur and Sathyawageeswar Subramanian are supported by the UKRI Future Leaders Fellowship MR/S031545/1.

\printbibliography

%%%%%%%%%%%%%%%%%%%%%%%%%%%%%%%%%%%%%%%%%%%%%%%%%%%%%%%%%%%%%%%%%%%%%%%%%%%%%%%%%%%%%%%%%%
% appendices
%%%%%%%%%%%%%%%%%%%%%%%%%%%%%%%%%%%%%%%%%%%%%%%%%%%%%%%%%%%%%%%%%%%%%%%%%%%%%%%%%%%%%%%%%%
\clearpage
\appendix
%%%%%%%%%%%%%%%%%%%%%%%%%%%%%%%%%%%%%%%%%%%%%%%%%%%%%%%%
\section{Creating block encodings or projected unitary encodings from the purified access oracle}
\label{app:bloc-enc}
We give an overview of the methods used by \cite{Gilyen2019Distributional} to obtain projected unitary encodings from purified access oracles in this appendix. Recall that following \cite[Definition 43]{Gilyen2018QuantumArithmetics}, we defined
    an $(\alpha, a, \epsilon)$ projected unitary encoding of an operator $A$ acting on $s$ qubits is a unitary $U$ acting on $a+s$ qubits, such that 
    \begin{equation}
        \norm{A - \alpha\Pi^\dagger U\tilde{\Pi}} \leq \epsilon,
    \end{equation}
    where the first register consists of ancillary qubits, $\Pi$ and $\tilde{\Pi}$ represent projections, i.e.\ $\Pi:=\ket{0}^{\otimes a}\otimes\id_s$ is an isometry mapping $(\C^2)^{\otimes s}\mapsto\mathrm{span}_{\C}\{\ket{0}^{\otimes a}\}\otimes (\C^2)^{\otimes s}$, and $\alpha,\epsilon\in(0,\infty)$. Below, we recall how to obtain such block encodings from purified access oracles to classical probability distributions and mixed states.

\subsection{Classical distributions}
In the case of a classical input distribution, the purified access oracle in \cref{eq:purified-access-class} can be turned into a block encoding for a matrix with singular values equal to the $\sqrt{p_j}$ as follows.

We choose $\Pi:=\displaystyle\sum_{i\in[n]}\id\otimes\ketbra{i}{i}\otimes\ketbra{i}{i}$, and $\tilde{\Pi}:=\ketbra{0}{0}\otimes\ketbra{0}{0}\otimes\id$, where each of the three registers is of dimension $n$. With $W=U\otimes\id$, we have that
\begin{equation}
    P = \Pi U \tilde{\Pi} = \sum_{i\in[n]} \sqrt{p_i}\ketbra{\phi_i ii}{00i}.
\end{equation}
The right hand side above represents the singular value decomposition (SVD) of a matrix $P$ with singular values $\sigma_i=\sqrt{p_i}$, left singular vectors $\ket{\phi_i i i}$ and right singular vectors $\ket{00i}$, i.e.
\begin{align*}
    P\ket{00i} &= \sqrt{p_i}\ket{\phi_i i i}\nonumber\\
    P^{\dagger}\ket{\phi_i i i } &= \sqrt{p_i}\ket{00i}.
\end{align*}
Hence we see that $U$ furnishes a $(1,\lceil\log n\rceil, 0)$ block encoding of $P$. 

\subsection{Arbitrary quantum density matrices}
Classical probability distributions correspond to the special case when $\rho$ is diagonal in the computational basis. For arbitrary density matrices, it is a little bit harder to create a projected unitary encoding that has the square roots of the eigenvalues of $\rho$, $\sqrt{p_i}$, as the singular values. Instead, \cite{Gilyen2019Distributional} give a construction which has singular values $\sqrt{\frac{p_i}{n}}$. 

Since we do not know the eigenbasis of $\rho$ beforehand, we define the projection operators $\Pi=\id\otimes\proj{0}\otimes\proj{0}$ and $\tilde{\Pi}=\proj{0}\otimes\proj{0}\otimes\id$. As before, the third register contains the ancilla, and the first two contain states of the system. We also need a unitary map $W$ that prepares the maximally entangled state on the two copies of the system register: 
\[
    \ket{0}\ket0\longmapsto\frac{1}{\sqrt{n}}\sum_{j=1}^n\ket j\ket j.
\] 
With these operators and the purified access oracle $U_\rho$ of \cref{eq:purified-access-quant}, we can define the unitary
\begin{equation*}
    U' = \left(\id\otimes U_\rho^{\dagger}\right)\left(W\otimes\id\right),
\end{equation*}
which gives rise to the following projected unitary encoding
\begin{equation*}
    \label{eq:gilyen_encoding}
    \Pi U'\tilde{\Pi} = \frac{1}{\sqrt{n}}\sum_{j=1}^n\ketbra{\phi_j'}{0}\otimes\ketbra{0}{0}\otimes\ketbra{0}{\psi_j},
\end{equation*}
where $\{\ket{\phi_j'}\}$ is a Schmidt basis for the first half of the bipartite maximally entangled state of dimension $n$.

This encoding has a leading normalisation factor of $\frac{1}{\sqrt{n}}$, which directly contributes a factor of $\sqrt n$ to the complexity of entropy estimation for quantum density matrices. 

\section{Implementing power functions of block encoded matrices}
\label{app:pow-funcs}
Polynomial functions of a matrix are defined through its singular value decomposition. An $n\times n$ matrix $P$ has $n$ real singular values $\sigma_j$, with a singular value decomposition in terms of its left and right singular vectors $\ket{v_j}$ and $\ket{w_j}$. An even or odd polynomial function $f$ of such a matrix is then defined as having the same singular vectors, but with the eigenvalues $f(\sigma_j)$, as follows
\begin{align}
\label{eqn:matrixFnDefn}
P&=\sum_{j=1}^n \lambda_j \ket{v_j}\bra{w_j};\nonumber\\
f(P)&=\begin{cases}
        \sum_{j=1}^n f(\sigma_j)\ket{w_j}\bra{w_j} & \text{when } f(x)=f(-x)\\
        \sum_{j=1}^n f(\sigma_j)\ket{v_j}\bra{w_j} & \text{when } f(x)=-f(-x)\\
    \end{cases}
\end{align}

The key trick in using QAE to estimate a functional $\varphi(\pbold):=\sum_i f(p_i)$ of an input vector $\pbold$ (in our case, a probability mass function) is to use the input purified access unitary that performs the map
\begin{equation*}
    U\ket{0^d}\ket{0^d} = \ket{\uppsi_{\pbold}} = \sum_{i=1}^n\sqrt{p_i}\ket{\phi_i}\ket{\psi_i}
\end{equation*}
where $d=\lceil\log n\rceil$, to construct a new unitary circuit which on a chosen, easy to prepare initial state performs a map of the type
\begin{equation*}
    W\ket{\uppsi_{\rho}}\ket{0}_{\text{flag}} =  \sum_{i=1}^n\sqrt{f(p_i)}\ket{\phi_i}\ket{\psi_i}\ket{0}_{\text{flag}} + \ket{\text{...}}\ket{1}_{\text{flag}},
\end{equation*}
so that value of the target functional $\varphi$ is encoded in the amplitude of the part of the output state marked by the $\ket0$ subspace of the flag register.

A projected unitary encoding $U$ of a matrix $P$ can be used to implement such smooth functions of the input matrix via polynomial approximations, with the following theorem.

\begin{theorem}[Theorem 56, \cite{Gilyen2018QuantumArithmetics}]
\label{thm:QSVT}
    Given an $(\alpha, a, \epsilon)$ block encoding $U$ of a Hermitian matrix $P$, for any degree $m$ polynomial $f(x)$ that satisfies $\forall~x\in[-1,1],~|f(x)|<\nicefrac12$, there exists a $(1, a+2, 4m\sqrt{\nicefrac{\epsilon}{\alpha}}+\delta)$ block encoding $U_f$ of $f(P/\alpha)$.  We can construct $U_p$ using $m$ applications of $U$ and $U^{\dagger}$, a single application of controlled-$U$, and $\O((a+1)m)$ additional $1$- and $2$-qubit gates. A description of the circuit of $U_f$ can be calculated in $\O(\textnormal{poly}(m, \log\nicefrac{1}{\delta}))$ time on a classical computer.
\end{theorem}

Using Theorem \ref{thm:QSVT}, we can implement $\epsilon$-approximate block encodings of power functions $P^c$ on the part of the singular value spectrum of $P$ that is contained in $[\delta,1]$ for $\delta>0$ by using polynomial approximations. The lower cutoff $\delta$ is necessary because power functions for non-integer exponents $c\in\R$ are not differentiable at $x=0$. On the other hand, monomials for $c=1,2,\ldots$ can be implemented exactly on the entire domain $[0,1]$.

% The (classical) overhead for finding the angles to be used in QSP or QSVT is given by Haah (or see section 3.1 pg 197, above Lemma 6, in the QSVT STOC paper), and is O(d^3 polylog(1/e)) for a degree d polynomial. So in our case we would incur some [(1/delta)log(d/eps)]^3 type cost. The 1/delta^3 seems significant...

We first note the following way \cite{Chakraborty2018,Gilyen2018QuantumArithmetics} of obtaining polynomial approximations of any desired degree for positive and negative power functions over a domain $[x_0-r-\delta, x_0+r+\delta]$ of radius $r\in(0,2]$ centred around a point $x_0\in[-1,1]$, with some wiggle room for the polynomial to vary, specified by the parameter $\delta\in(0,r]$.

\parahead{Positive Power functions: } Consider $f(x)=x^{c}$ for $c>0$. The Taylor series expansion of $f$ around $x_0=1$ 
\begin{align}
    f(1+x) &= (1+x)^c\nonumber\\
        &=1+\sum_{k=1}^\infty \binom{c}{k}x^k
\end{align}
converges $\forall x\in[-1,1]$, where
\[
    \binom{c}{k} := \frac{c(c-1)(c-2)\ldots(c-k+1)}{k!}.
\]
Notice that we have
\begin{align*}
    1+\sum_{k=1}^\infty \Abs{\binom{c}{k}}(1-\delta + \delta)^k &= 1+\sum_{k=1}^\infty \Abs{\binom{c}{k}}\\
        &=1-\sum_{k=1}^\infty \binom{c}{k}(-1)^k\\
        &=2-\sum_{k=0}^\infty \binom{c}{k}(-1)^k\\
        &=2-f(1-1)\\
        &=2.
\end{align*}
Specific to our purpose, this means we can choose $x_0=1$, $r=1-\sqrt{\beta}/2$, $\delta=\sqrt{\beta}/2$, and the normalisation factor $B=2$ for implementing positive power functions of block encodings of $P$ with singular values $\pbold=(\sqrt{p_1},\ldots,\sqrt{p_n})$ on the domain $[\sqrt{\beta},1]$ corresponding to the heavy elements $i\in B_\beta$ with probability masses $p_i\geq\beta$. 

\parahead{Negative Power functions: }Consider $f(x)=x^{-c}$ for $c>0$. The Taylor series expansion of $f$ around $x_0=1$ 
\begin{align}
    f(1+x) &= (1+x)^{-c}\nonumber\\
        &=1+\sum_{k=1}^\infty \binom{-c}{k}x^k
\end{align}
converges $\forall x\in[-1,1]$, where
\[
    \binom{-c}{k} := \frac{-c(-c-1)(-c-2)\ldots(-c-k+1)}{k!}.
\]
With $\delta':=\frac{\delta}{2\max(1,c)}$ notice that we have
\begin{align*}
    1+\sum_{k=1}^\infty \Abs{\binom{-c}{k}}(r + \delta')^k &= 1+\sum_{k=1}^\infty \binom{-c}{k}(-r - \delta')^k\\
        &=(1-r-\delta')^{-c}\\
        &=(\delta-\delta')^{-c}\\
        &=\delta^{-c}(1-\frac{\delta'}{\delta})^{-c}\\
        &=\delta^{-c}(1-\frac{1}{2\max(1,c)})^{-c}\\
        &=2\delta^{-c}.
\end{align*}
If we choose to normalise the original function to $\frac{\delta^{c}}{2}x^{-c}$, the above calculation shows that we can choose $x_0=1$, $r=1-\sqrt{\beta}/2$, $\delta=\sqrt{\beta}/2$, and the normalisation factor $B=1$ for implementing negative power functions over $[\sqrt{\beta},1]$. Since the case of negative power functions is more illustrative than positive ones, we state this formally below. A similar statement holds for the positive case.

\begin{lemma}[Corollary 67, \cite{Gilyen2018QuantumArithmetics}]
\label{lem:powerfunctions}
    Given a $(\alpha, a, 0)$ unitary block encoding $U$ of a matrix with singular value decomposition $P=\sum_i \sigma_i \ketbra{v}{w}$, and an even polynomial $f_c(x)$ on $[-1,1]$ that approximates the negative power function $x^{-c}$ for $c>0$ such that 
    \begin{align*}
        \Abs{f_c(x)-\frac{\delta^c}{2}x^{-c}}~~&\leq~~ \epsilon &\forall x\in[\delta,1]\\
        \Abs{f_c(x)}~~&\leq~~ 1 &\forall x\in[-1,1]\\
        m=\deg{f} &= \O\left(\frac{\max(1,c)}{\delta}\log \frac 1\epsilon\right),
    \end{align*}
    we can implement a $(\nicefrac{2}{\delta^c}, a+2, \epsilon)$ block encoding $U_f$ of the matrix polynomial $f_c(P)=\sum_i f_c(\sigma_i) \ketbra{w}{w}$ using m  
    applications of $U$ and $U^{\dagger}$, a single application of controlled-$U$, and $\O(ma)$ additional one- and two-qubit gates. Furthermore, a description of the quantum circuit $U_f$ can be computed classically in time $\O(\poly(m,\log\nicefrac{1}{\epsilon}))$.
\end{lemma}

\section{Quantum phase estimation and singular value estimation}
\label{app:phase-est}

We would like to use Quantum Phase Estimation as a subroutine to separately flag the subspaces of heavy and light elements. In essence we want to perform the map in \cref{eq:qpe-qsve}, i.e.\
\begin{equation}
    \sum_{i\in[n]} \sqrt{p_i}\ket{\phi_i}\ket{i}\otimes \ket{0^{m}} \mapsto \sum_{i\in[n]} \sqrt{p_i}\ket{\phi_i}\ket{i} \ket{q_i},
\end{equation}
where $\Abs{\sqrt{p_i}-q_i}\leq 2^{-(m+1)} =: \epsilon$, and $m$ is the number of bits of precision.

Recall that we have a block encoding $U$ of a matrix $P$ that represents our input distribution, where $P$ has the singular value decomposition
\begin{equation}
    P = \tilde{\Pi}U\Pi = \sum_{i\in[n]} \sqrt{p_i} \ketbra{\phi_i i i}{00i}.
\end{equation}

Functions of $P$ defined by even or odd polynomials $f$ or $\tilde{f}$ respectively acting on the singular values then have the form
\begin{align}
    f(P) &:=  \sum_{i\in[n]} f(\sqrt{p_i}) \ketbra{00i}{00i}\nonumber\\
    f(P) &:=  \sum_{i\in[n]} f(\sqrt{p_i}) \ketbra{\phi_i i i}{00i}.
\end{align}

In principle, we can simply use the standard textbook version of the quantum phase estimation algorithm (QPE) which requires controlled-$U$ operators and the quantum fourier transform (QFT) in order to estimate the phases $\theta_j\in[0,1)$ of the eigenvalues $\lambda_j=\ee^{2\pi \ii\theta_j}$ of $U$. We consider $e^{2\pi \ii Pt}$ as the input unitary, which can be implemented using the block encoding of $P$ via Hamiltonian simulation, with query complexity to $U$ and $U^{\dagger}$ bounded by 
\[
    \O\left(t+\frac{\log\nicefrac{1}{\epsilon}}{\log\log\nicefrac{1}{\epsilon}}\right)
\]
where $\epsilon$ is a precision parameter defined by $\forall j\in[n],~\Abs{\ee^{2\pi\ii\sqrt{p_j}} - \lambda_j}\leq \epsilon$ \cite{Low2016,Gilyen2018QuantumArithmetics}. This condition translates to $\Abs{\theta_j - \sqrt{p_j}}\leq \frac{1}{\pi}\arcsin\frac{\sqrt{\epsilon}}{2}\leq \sqrt{\epsilon}$. This error is benign as far as we are concerned: we can choose it to be of order $\nicefrac{1}{n}$ or even $\nicefrac{1}{n^2}$ while incurring only a additive logarithmic overhead in the complexity. Indeed, we shall choose it to be inverse polynomial in $n$, so that the subsequent step that uses QPE to estimate $\theta_j$ still behaves as we expect it to --- it will produce an estimate of zero whenever $\theta^2_j\leq\beta \iff {p_j}<\beta$.

While the above explanation is the high level intuition, in practice things can be a bit more delicate, and we use technique of quantum singular value estimation (QSVE). Since $P$ is not Hermitian, we consider a symmetrised version of it defined by $\hat{P}=\ketbra{0}{1}\otimes P + \ketbra{1}{0}\otimes P^{\dagger}$, which has eigenvectors $\ket{0}\otimes\ket{\phi_i i i} + \ket{1}\otimes\ket{00i}$ and eigenvalues $\sqrt{p_i}$. The problem then is to perform the map in \cref{eq:qpe-qsve} using the block encoding of $\hat{P}$, which in turn can easily be constructed using the block encoding of $P$. This matches the problem addressed in \cite{Kerenidis2017,Chakraborty2018,Gilyen2018QuantumArithmetics}, and the complexity is essentially $\widetilde{O}\left(\nicefrac{1}{\epsilon}\right)$ where $\epsilon$ is the precision to which we would like to estimate the singular values. 

We can then choose $m=\log\sqrt{\nicefrac{1}{\beta}}$, and the query complexity of the QSVE subroutine becomes 
\[
    \O\left(\frac{1}{\epsilon}\right) = \O\left(\frac{1}{\sqrt{\beta}}\right).
\]

\section{Quantum amplitude estimation}
\label{app:amp-est}
Quantum amplitude estimation (QAE) is a technique wherein quantum phase estimation (QPE) is used to estimate the amplitude of a certain basis state (more generally, of any state about which we can perform a reflection operation) in a superposition produced by applying a unitary operation $U$ to a given input state. The QPE algorithm is applied to estimate the eigenvalues of the Grover iterate constructed from the input unitary. We also note that the most modern methods of QAE do not rely on QPE or the quantum fourier transform in an essential way \cite{SuzukiQAE-QFT,AaronsonQAE-QFT}.
\begin{theorem}[\cite{Brassard2002}, Theorem 12]
    Given a unitary $U$ with the action
    \begin{equation*}
        U\ket 0\ket 0 = \sqrt{p}\ket 0 \ket\phi + \ket{\perp},
    \end{equation*}
    where $\ket\phi$ is a normalised state on the system register, and $\left(\proj{0}\otimes\id\right)\ket{\perp}=0$, the quantum amplitude estimation algorithm outputs $\tilde p\in[0,1]$ satisfying
    \begin{equation*}
        \Abs{p-\tilde p}\leq \frac{2\pi\sqrt{p(1-p)}}{M} + \frac{\pi^2}{M^2},
    \end{equation*}
    with a success probability at least $8/\pi^2$, making $M$ uses of $U$ and $U^\dagger$. 
\end{theorem}

To get an approximation $\tilde p$ that is correct to a constant additive precision $\epsilon\in(0,1/2)$, we can choose $M=\lceil 2\pi\left(\frac{2\sqrt{p}}{\epsilon} + \frac{1}{\sqrt{\epsilon}}\right) \rceil = \Theta\left(\frac{\sqrt{p}}{\epsilon} + \frac{1}{\sqrt{\epsilon}}\right)$, and hence with a complexity of $\Theta(1/\epsilon)$ we can estimate $p$ to additive precision $\epsilon$.

%%%%%%%%%%%%%%%%%%%%%%%%%%%%%%%%%%%%%%%%%%%%%%%%%%%%%%%%%%%%%%%%%%%%%%%%%%%%%%%%%%%%%%%%%%
\end{document}